\documentclass[a4paper,notitlepage,fleqn,12pt]{article}
 \usepackage[usenames,dvipsnames]{color}

\usepackage[tbtags]{amsmath}
\usepackage{amsfonts}
\usepackage{amsthm}
\usepackage{color}
\usepackage{amsfonts}
\usepackage{graphicx,amssymb}
\usepackage{multirow}
\usepackage{rotating}
\usepackage{mathrsfs}
\usepackage{epstopdf}
\usepackage{afterpage}
\usepackage{natbib}
\usepackage{tabularx}
\usepackage{xr}
\usepackage{subfigure}
\usepackage{caption}

\newtheorem{lemma}{Lemma}
\newtheorem{thm}{Theorem}

\newtheorem{assum}{Assumption}
\newtheorem{remark}{Remark}

\DeclareMathOperator*{\sgn}{sgn}
\DeclareMathOperator*{\diag}{diag}
\DeclareMathOperator*{\cov}{cov}
\DeclareMathOperator*{\var}{var}
\newcommand{\bm}[1]{\mbox{\boldmath{$#1$}}}

\numberwithin{equation}{section}

\title{Asymmetric linear double autoregression}
\author{Songhua Tan and Qianqian Zhu\\
\emph{School of Statistics and Management}\\
\emph{Shanghai University of Finance and Economics}
}
\date{}

\begin{document}
	\maketitle
	\begin{abstract}
		This paper proposes the asymmetric linear double autoregression, which jointly models the conditional mean and conditional heteroscedasticity characterized by asymmetric effects. 
		A sufficient condition is established for the existence of a strictly stationary solution. 
		With a quasi-maximum likelihood estimation (QMLE) procedure introduced, a Bayesian information criterion (BIC) and its modified version are proposed for model selection. 
		To detect asymmetric effects in the volatility, the Wald, Lagrange multiplier and quasi-likelihood ratio test statistics are put forward, and their limiting distributions are established under both null and local alternative hypotheses. 
		Moreover, a mixed portmanteau test is constructed to check the adequacy of the fitted model. 
		All asymptotic properties of inference tools including QMLE, BICs, asymmetric tests and the mixed portmanteau test, are established without any moment condition on the data process, which makes the new model and its inference tools applicable for heavy-tailed data. 
		Simulation studies indicate that the proposed methods perform well in finite samples, 
		and an empirical application to S\&P500 Index illustrates the usefulness of the new model. 
	\end{abstract}
	
	\newpage
	
	\section{Introduction} \label{Introduction}
	
	Volatility clustering is a major feature of financial time series, and the ability to forecast volatility is of vital importance for the pricing and risk management of financial assets. To capture the time-varying volatility, many conditional heteroscedastic models are proposed, and among them, the autoregressive conditional heteroscedastic (ARCH) and the generalized autoregressive conditional heteroscedastic (GARCH) models \citep{Engle1982,Bollerslev1986} are very successful specifications.  
	However, empirical facts indicate that the autocorrelation and volatility dynamics usually coexist in time series; see, for example, the daily returns of NASDAQ Composite Index in \cite{kuester2006value} and the weekly or monthly returns of S\&P500 index in \cite{linton2005estimating}. 
	As a result, to better capture the volatility dynamics in the presence of data autocorrelations, it is necessary to jointly model the conditional mean and volatility \citep{Li_Ling_McAleer2002}.   
	The autoregressive moving average models with GARCH errors (ARMA-GARCH) and double autoregressive (DAR) models are popular specifications for this purpose. 
	
	In financial applications, ARMA-GARCH models are commonly used to fit return series \citep{Francq2019}. 
	Many researchers have studied the estimation for ARMA-GARCH models; see among others, \cite{Francq_Zakoian2004} and \cite{Zhu_Ling2011}. \cite{Francq_Zakoian2004} has shown that, a finite fourth moment on the data process is required to establish the asymptotic normality of the Gaussian quasi-maximum likelihood estimator (QMLE) for the ARMA-GARCH model. This largely narrows down the application of ARMA-GARCH models since the heavy-taildeness is very common for financial data.    
	Meanwhile, the DAR model proposed by \cite{Ling2007} has recently attracted growing attention. The DAR model  of order $p$ is defined as 
	\begin{equation}\label{DAR}
	y_{t}=\sum_{i=1}^{p} \phi_{i} y_{t-i}+\varepsilon_{t}\sqrt{\omega+\sum_{i=1}^{p}\beta_{i} y_{t-i}^2},
	\end{equation}
	where $\omega>0$, $\beta_i\geq 0$ for $1\leq i\leq p$, and $\{\varepsilon_t\}$ are independent and identically distributed ($i.i.d.$) innovations with zero mean and unit variance. In contrast to the ARMA-GARCH model, the Gaussian QMLE of model \eqref{DAR} is asymptotically normal provided that $y_t$ has a fractional moment \citep{Ling2007}. 
	This important property makes model \eqref{DAR} suitable for handling heavy-tailed data in application. 
	In addition, the DAR model with order $p=1$ can still be stationary even if $|\phi_1|>1$ or $\beta_1>1$, thus it enjoys a larger parameter space than conventional AR and AR-ARCH models.
	
	Many variants of DAR models have been widely proposed and studied, such as the threshold DAR \citep{Li2016_TDAR}, the mixture DAR \citep{Li_Zhu_Liu_Li2017}, the linear DAR \citep{Zhu2018_LDAR} and the augmented DAR \citep{Jiang2020} models.  
	Specifically, the linear DAR model of order $p$ has the form of  
	\begin{equation}\label{LDAR}
	y_{t}=\sum_{i=1}^{p} \phi_{i} y_{t-i}+\varepsilon_{t}\left(\omega+\sum_{i=1}^{p}\beta_{i} |y_{t-i}|\right),
	\end{equation}
	where the innovations $\{\varepsilon_t\}$ and parameters are defined as in model \eqref{DAR}.
	Model \eqref{LDAR} assumes that the conditional standard deviation rather than the conditional variance of $y_t$ is in a linear structure, which can lead to more robust inference than model \eqref{DAR}; see \cite{Taylor2008} and \cite{Zhu2018_LDAR}. 
	As shown by \cite{Zhu2018_LDAR}, the linear DAR model has a larger parameter space than conventional AR and AR-ARCH models as for DAR models. 
	Moreover, the asymptotic normality of the Gaussian QMLE can also be established for model \eqref{LDAR} without any moment restrictions on $y_t$; see \cite{liuhua2020QMLI}. 
	As a result, the linear DAR model enjoys the important property of DAR models and hence can also be used to fit heavy-tailed data. 
	
	It is well known that financial time series are usually characterized by asymmetry (leverage) effects, in the sense that the volatility of financial returns tends to be higher after a decrease than an equal increase. The leverage effect was documented by many authors as a stylized fact of stock returns; see, for example, \cite{black1976studies}, \cite{rabemananjara1993thresholdARCH} and \cite{Francq2013_nonstationary_agarch}. 
	To account for the leverage phenomenon, many variants of classical GARCH models are introduced and studied, such as the exponential GARCH \citep{nelson1991conditional}, the threshold GARCH \citep{zakoian1994threshold} and the power GARCH \citep{Pan_Wang_Tong2008} models. 
	Specifically, \cite{engle1993measuring} defined the news impact curve to measure how new information is incorporated into volatility estimates, and based on this curve they provided diagnostic tests to detect asymmetric effects of news on volatility. 
	However, limited literatures investigate the leverage effect in the presence of the conditional mean structure, even rare in the framework of DAR models. To fill this gap, we propose an asymmetric linear DAR model, which can be regarded as a modification of the linear DAR model along the lines of the threshold GARCH model. 
	Hopefully, the new model can preserve the advantages of DAR type models in handling with heavy-tailed data and meanwhile be able to capture the asymmetric effect successfully. 
	The main contributions of this paper are listed as follows. 
	
	First, Section \ref{ALDAR} introduces an asymmetric linear DAR model to capture leverage effects, where the coefficients for positive and negative parts of $y_t$ in the conditional standard deviation can be different. 
	We establish a sufficient condition for the strict stationarity and ergodicity of the new model by showing that the Markov chain $\bm Y_t=(y_t, \ldots, y_{t-p+1})^{\prime}$ is $\nu_p$-irreducible and satisfies Tweedie's drift criterion \citep{Tweedie1983}; see also \cite{Zhu2018_LDAR}. It is shown that the stationary region of the proposed process depends on the moment condition of innovations and the degree of asymmetry. Moreover, the proposed process can still be stationary even if some AR coefficient is greater than one, leading to a large stationary region as for the DAR and linear DAR models.

	Secondly, Section \ref{estimation} proposes the Gaussian QMLE for the new model and establishes its consistency and asymptotic normality. 
	Particularly, the consistency is carefully considered to avoid the identification problem due to the coexistence of the positive and negative parts of $y_t$, and the asymptotic normality is established without any moment condition on $y_t$. Hence, the new model can be used to fit heavy-tailed data. 
	Moreover, based on the QMLE, a Bayesian information criterion (BIC) is proposed for model selection, and a modified BIC is introduced to improve the finite-sample performance. It is shown that both BICs enjoy the selection consistency without any moment condition on the process $\{y_t\}$, and the modified BIC usually outperforms the unmodified BIC especially for small and moderate samples. 
	As a result, the first two stages of Box-Jenkins' procedure including model specification and estimation are constructed for the new model in a robust way, which facilitates its application to financial time series. 
	
	Thirdly, to detect the asymmetry in volatility, the Wald, Lagrange multiplier (LM) and quasi-likelihood ratio (QLR) tests are constructed in Section \ref{test}. Under the null and local alternative hypotheses, the Wald and LM test statistics are shown to have the same limiting distributions, while the QLR test statistic converges to weighted sums of $i.i.d.$ central and noncentral chi-squared random variables, respectively. 
	It is noteworthy that, to show the asymptotic distributions of three test statistics under local alternatives, we need to verify the local asymptotic normality (LAN) of the proposed model. 
	However, without the normal assumption on the innovation term, the Le Cam's third lemma cannot be employed to show the LAN property \citep{Vaart2000_LAN}. Alternatively, we show that the new model satisfies the LAN in a direct way; see also \cite{Jiang2020}. 
	Moreover, we conduct simulation experiments to compare the local power of all three tests in finite samples, and simulation results further support the theoretical findings. 
	
	Finally, we investigate diagnostic checking for fitted models using a mixed portmanteau test in Section \ref{checking}. 
	The portmanteau test for pure mean models \citep{Ljung_Box1978} is constructed using the sample autocorrelation functions (ACFs) of residuals, while that for volatility models \citep{Li_Li2008} employs the ACFs of squared or absolute residuals. 
	As validated by simulation findings of \cite{Li_Li2005}, the portmanteau test based on absolute residuals is less powerful than that based on squared residuals under heavy-tailed situations. As a result, 
	this paper proposes a mixed portmanteau test \citep{Wong2005mixed_portmanteau} using the ACFs of residuals and absolute residuals to detect inadequacy in the conditional mean and volatility of the fitted new model. The joint limiting distribution for ACFs of residuals and absolute residuals is established without any moment condition on the process $\{y_t\}$. Simulation results validate that the mixed test can detect inadequacy of fitted models due to either the conditional mean or the volatility. Therefore, as the last stage of Box-Jenkins' procedure, the diagnostic checking tool is successfully constructed in a robust way as well. 

	In addition, Section \ref{simulation} conducts simulation studies to evaluate the finite-sample performance of all inference tools for the proposed model. 
	Section \ref{real_data} illustrates the usefulness of the new model by analyzing the S\&P500 Index and demonstrates the forecasting superiority over its counterparts especially for the skewed and heavy-tailed data. 
	The conclusion and discussion appear in Section \ref{conclusion}. 
	All technical details are relegated to the Appendix. 
	Throughout the paper, $\mathcal{N}$ denotes the integer, $\rightarrow_p$ and $\rightarrow_{\mathcal{L}}$ denote the convergences in probability and in distribution, respectively, and $o_p(1)$ denotes a sequence of random variables converging to zero in probability.
	
	\section{Asymmetric linear double autoregression} \label{ALDAR}
	
	Consider the asymmetric linear double autoregressive (DAR) model  of order $p$,
	\begin{equation}
	y_{t}=\sum_{i=1}^{p} \alpha_{i} y_{t-i}+\eta_{t}\left(\omega+\sum_{i=1}^{p}\left(\beta_{i+} y_{t-i}^{+}-\beta_{i-} y_{t-i}^{-}\right)\right),
	\label{model}
	\end{equation}
	where $\omega>0,\beta_{i+},\beta_{i-}\geq0$ for $1\leq i\leq p$, $y^+_t = \max\left\lbrace 0,y_t\right\rbrace$ and $y^-_t = \min\left\lbrace 0,y_t\right\rbrace$ are positive and negative parts of $\left\lbrace y_t\right\rbrace$, respectively, and $\{\eta_{t}\}$ is a sequence of $i.i.d.$ random variables with mean zero and variance one.
	The asymmetric linear DAR model in \eqref{model} is an extension of the linear DAR model \citep{Zhu2018_LDAR} along the lines of the threshold GARCH model \citep{zakoian1994threshold}.
	Although the linear DAR model can be extended to allow for asymmetries in both the conditional mean and conditional heteroscedasticity, this paper focuses on model \eqref{model} to take account for the asymmetry in volatilities. The real example of stock index returns in Section \ref{real_data} provides evidence for this motivation.
	
	For general distributions of $\eta_{t}$, it is difficult to derive a necessary and sufficient condition for the strict stationarity due to the nonlinearity of model \eqref{model}; see also \cite{Li2016_TDAR} and \cite{Zhu2018_LDAR}.
	Alternatively, a sufficient condition is provided below.
	\begin{assum}
		The density function of $\eta_{t}$ is continuous and positive everywhere on $\mathbb{R}$, and $E(|\eta_{t}|^\kappa)<\infty$ for some $\kappa>0$.
		\label{assum1}
	\end{assum}
	\begin{thm}
		Under Assumption \ref{assum1}, if either of the following conditions holds: \\
		(i) for $0<\kappa\leq 1$, $\sum^p_{i=1}\max\left\lbrace E\left( |\alpha_i-\beta_{i-}\eta_t|^\kappa \right), E\left( |\alpha_i+\beta_{i+}\eta_t|^\kappa \right) \right\rbrace<1$; \\
		(ii) for $\kappa \in \{2,3,4,\ldots\}$, $E\left[\left( \sum_{i=1}^{p}\max\left\lbrace|\alpha_i+\beta_{i+}\eta_{t}|,|\alpha_i-\beta_{i-}\eta_{t}|\right\rbrace \right)^\kappa\right]<1$;\\
		then there exists a strictly stationary solution $\{ y_t \}$ to model \eqref{model}, and this solution is unique and geometrically ergodic with $E\left(|y_t|^\kappa \right)<\infty$.
		\label{thm1Stationarity}
	\end{thm}
	
	The stationarity region in Theorem \ref{thm1Stationarity} depends on the distribution of $\eta_{t}$ and implies a moment condition on $y_t$.
	In addition, when $\eta_{t}$ has a symmetric distribution and the asymmetric linear DAR model reduces to a linear DAR model, that is $\beta_{i-}=\beta_{i+}=\beta_{i}$, then it simplifies to $\sum^p_{i=1}E\left(|\alpha_i+\beta_{i}\eta_t|^\kappa \right)<1$ for $0<\kappa\leq 1$, and to $E\left[\sum^p_{i=1}(|\alpha_i|+\beta_{i}|\eta_t|)\right]^\kappa<1$ for $\kappa \in \{2,3,4,\ldots\}$.
	Since the stationarity region of model \eqref{model} is at least three-dimensional, for illustration, we provide the stationarity regions of model \eqref{model} of order one, and consider $\beta_{1+}=d\beta_{1-}$ with the constant $d$ being different positive values. 
	Figure \ref{fig1}(a) indicates that model \eqref{model} of order one can be stationary if $|\alpha_1|\geq 1$, hence model \eqref{model} preserves a large parameter space as DAR and linear DAR models.
	As shown in Figure \ref{fig1}(b), a larger value of $\kappa$ in Theorem \ref{thm1Stationarity} leads to a higher moment of $y_t$, and hence results in a narrower stationarity region.
	Moreover, Figure \ref{fig1}(c) shows that the stationarity region gets smaller as the asymmetry in volatilities becomes greater.
	
	\begin{remark}\label{remark-model order}
		The order of model \eqref{model} can be different for the conditional mean and volatility, that is, we can consider the asymmetric linear DAR model of order $(p_1,p_2)$ 
		\[y_{t}=\sum_{i=1}^{p_1} \alpha_{i} y_{t-i}+\eta_{t}\left(\omega+\sum_{i=1}^{p_2}\left(\beta_{i+} y_{t-i}^{+}-\beta_{i-} y_{t-i}^{-}\right)\right),\] 
		where $p_1$ and $p_2$ are positive integers. 
		For the strict stationarity of this general model setting, Theorem \ref{thm1Stationarity} still holds by letting $p=\max\{p_1,p_2\}$ with $\alpha_{i}=0$ for $i>p_1$ and $\beta_{i+}=\beta_{i-}=0$ for $i>p_2$. 	
		However, if $p_1>p_2$, the conditional scale structure $\omega+\sum_{i=1}^{p_2}\left(\beta_{i+} y_{t-i}^{+}-\beta_{i-} y_{t-i}^{-}\right)$ cannot be used to reduce the moment condition on $y_{t}$ for showing the asymptotic normality of the quasi-maximum likelihood estimator in Section \ref{estimation}. 
		To establish the asymptotic normality without any moment condition on $y_{t}$, other estimation methods such as the self-weighted approach \citep{Ling2005} should be considered. We leave this extension for future research, and consider the same order $p=p_1=p_2$ in this paper. 
	\end{remark}	 
	
	\section{Model estimation} \label{estimation}
	
	\subsection{Quasi-maximum likelihood estimation} \label{QMLE}
	Let $\bm \theta = (\bm \alpha^\prime,\bm \beta^\prime)^{'}$ be the parameter vector of model \eqref{model}, where $\bm{\alpha} = (\alpha_1,\alpha_2,\ldots,\alpha_p)^{'}$, $\bm \beta = (\omega,\bm \beta^\prime_+,\bm \beta^\prime_-)^\prime$ with $\bm \beta_+ = (\beta_{1+},\beta_{2+},\ldots,\beta_{p+})^\prime$ and $\bm \beta_- = (\beta_{1-},\beta_{2-},\ldots,\beta_{p-})^\prime$.
	Denote the true parameter vector by $\bm \theta_0 = (\bm{\alpha}^\prime_{0},\bm{\beta}^\prime_{0})^\prime$ and the parameter space by  $\Theta$, where $\Theta$ is a compact subset of $\mathbb{R}^{p} \times \mathbb{R}_{+}^{2p+1}$ with $\mathbb{R}_{+}=(0,\infty)$.
	
	Let $\bm Y_{t} = (y_t,\ldots,y_{t-p+1})^{\prime}$ and $\bm X_t =(1,\bm Y^{\prime}_{t+},-\bm Y^{\prime}_{t-})^{\prime}$, where $\bm Y_{t+} = (y_t^{+},\ldots,y_{t-p+1}^{+})^{\prime}$ and $\bm Y_{t-} = (y_t^{-},\ldots,y_{t-p+1}^{-})^{\prime}$.
	The conditional log-likelihood function (ignoring a constant) can be written as
	\begin{equation}
	L_n(\bm \theta) = \sum_{t=p+1}^{n}\ell_t(\bm \theta) \quad \text{and} \quad \ell_{t}(\bm \theta)=-\ln \left(\bm \beta^{\prime} \bm X_{t-1}\right)-\dfrac{\left(y_{t}-\bm \alpha^{\prime} \bm Y_{t-1}\right)^{2}}{2\left(\bm{\beta}^{\prime} \bm{X}_{t-1}\right)^{2}}.
	\label{likelihood function}
	\end{equation}
	Then the quasi-maximum likelihood estimator (QMLE) of $\bm \theta_0$ can be defined as
	\begin{equation}
	\label{optimization}
	\widehat{\bm \theta}_n=\arg \max_{\bm \theta \in \Theta} L_n(\bm \theta).
	\end{equation}
	
	\begin{assum}\label{assum_y_t}
		$\left\lbrace y_t:t\in\mathcal{N}\right\rbrace $ is strictly stationary and ergodic with $E(|y_t|^\kappa)<\infty$ for some $\kappa>0$.
	\end{assum}
	
	\begin{assum}\label{assum_eta_t}
		The density function of $\eta_{t}$ is continuous and positive everywhere on $\mathbb{R}$. 
	\end{assum}
	
	\begin{assum}
		The parameter space $\Theta$ is compact with $\underline{\omega}\leq \omega \leq \bar{\omega}$, $\underline{\beta}\leq \beta_{i-},\beta_{i+} \leq \bar{\beta}$ for $i=1,\ldots,p$, where $\underline{w}, \bar{w}, \underline{\beta}, \bar{\beta}$ are some positive constants. The true parameter vector $\bm\theta_0$ is an interior point in $\Theta$.
		\label{assumption_compact}
	\end{assum}
	For the strict stationarity of $\{y_t\}$ in Assumption \ref{assum_y_t}, a sufficient condition is given in Theorem \ref{thm1Stationarity}. 
	Assumption \ref{assum_eta_t} is imposed for identifying the unique maximizer of $E[\ell_t(\bm \theta)]$ at $\bm \theta_0$; see also \cite{francq2012qml}.
	Assumption \ref{assumption_compact} is required to ensure the log-likelihood function, score function and information
	matrix to be bounded without any moment restrictions on $y_t$; see also \cite{Ling2007}.
	As a result, the model based on the QMLE can be applied to heavy-tailed data.
	
	Let $\kappa_1=E(\eta_t^3)$ and $\kappa_2=E(\eta_t^4)-1$. Define the $(3p+1)\times(3p+1)$ matrices
	\[
	\Omega= E\left[\dfrac{\partial \ell_{t}\left(\bm \theta_{0}\right)}{\partial \bm \theta} \dfrac{\partial \ell_{t}\left(\bm \theta_{0}\right)}{\partial \bm \theta^{\prime}}\right]=E\left(\begin{array}{cc}
	{\dfrac{\bm Y_{t-1} \bm Y_{t-1}^{\prime}}{\left(\bm \beta_{0}^{\prime} \bm X_{t-1}\right)^{2}}} & {\dfrac{\kappa_1\bm Y_{t-1} \bm X_{t-1}^{\prime} }{\left(\bm \beta_{0}^{\prime} \bm X_{t-1}\right)^{2}}} \\
	{\dfrac{\kappa_1\bm X_{t-1} \bm Y_{t-1}^{\prime}}{\left(\bm \beta_{0}^{\prime} \bm X_{t-1}\right)^{2}}} & {\dfrac{\kappa_2\bm X_{t-1} \bm X_{t-1}^\prime}{\left(\bm \beta_{0}^{\prime} \bm X_{t-1}\right)^{2}}}
	\end{array}\right),
	\]
	and
	\begin{equation}\label{Sigma}
	\Sigma = -E\left[ \dfrac{\partial^2 \ell_t(\bm \theta_0)}{\partial \bm \theta \partial\bm \theta^\prime} \right] =\diag\left\{E\left[\frac{\bm Y_{t-1} \bm Y_{t-1}^{\prime}}{(\bm \beta_{0}^{\prime} \bm X_{t-1})^2}\right], E\left[\frac{2 \bm X_{t-1} \bm X_{t-1}^{\prime}}{\left(\bm{\beta}_{0}^{\prime} \bm{X}_{t-1} \right)^{2}}\right]\right\}.
	\end{equation}
	
	\begin{thm}\label{thm2QMLE}
		Suppose that Assumptions \ref{assum_y_t}--\ref{assumption_compact} hold. Then, 
		\begin{itemize}
			\item[(i)] $\widehat{\bm \theta}_{n} \rightarrow_{p} \bm \theta_{0}$ as $n\to \infty$; 
			\item[(ii)] furthermore, if $E(\eta^4_{t}) < \infty$ and the matrix $D=\left(\begin{matrix}
			1 & \kappa_1\\
			\kappa_1& \kappa_2
			\end{matrix}
			\right)$ is positive definite, then $\sqrt{n}(\widehat{\bm\theta}_{n}-\bm\theta_{0}) \rightarrow_{\mathcal{L}} N\left(\bm 0, \Xi\right)$ as $n\to \infty$, where $\Xi=\Sigma^{-1}\Omega\Sigma^{-1}$.
		\end{itemize}
	\end{thm}
	Note that the positive definiteness of $D$ is satisfied for continuous $\eta_t$ with $E(\eta^4_{t}) < \infty$; see \cite{Jiang2020}.
	If $\eta_{t}$ is normal, then $\kappa_1=0, \kappa_2=2$ and $\Omega=\Sigma$, thus the QMLE reduces to the MLE and its asymptotics in Theorem \ref{thm2QMLE} can be simplified to $\sqrt{n}(\widehat{\bm \theta}_{n}-\bm \theta_{0}) \rightarrow_{\mathcal{L}} N\left(\bm 0, \Sigma^{-1}\right)$ as $n\to \infty$.
	To calculate the asymptotic covariance of $\widehat{\bm \theta}_{n}$, we use sample averages to replace matrices $\Omega$ and $\Sigma$, and the QMLE $\widehat{\bm \theta}_{n}$ to replace $\bm\theta_{0}$.
	
	\subsection{Model selection} \label{BIC}
	This subsection considers the selection of order $p$ for model \eqref{model} in practice.
	We first introduce the Bayesian information criterion (BIC) below to select the order $p$,
	\begin{equation}\label{BIC1}
	\text{BIC}_1(p) = -2L_n(\widehat{\bm\theta}_{n}^{p}) + (3p+1)\ln(n-p),
	\end{equation}
	where $\widehat{\bm\theta}_{n}^{p}$ is the QMLE when the order is set to $p$, and $L_n(\widehat{\bm\theta}_{n}^{p})$ is the log-likelihood evaluated at $\widehat{\bm\theta}_{n}^{p}$. 
	However, model \eqref{model} is fitted by the Gaussian QMLE, hence the model misspecification should be considered in deriving the asymptotic expansion of the Bayesian principle, which leads to the modified BIC below
	\begin{equation}\label{BIC2}
	\text{BIC}_2(p) = -2L_n(\widehat{\bm\theta}_{n}^{p}) + (3p+1)\ln\left(\dfrac{n-p}{2\pi}\right)+\ln(\det(\widehat{\Sigma}^{p})),
	\end{equation}
	where $\widehat{\Sigma}^{p}$ is a consistent estimator of $\Sigma$ defined as in \eqref{Sigma} at order $p$, and $\det(\widehat{\Sigma}^{p})$ is its determinant. 
	In practice, $\widehat{\Sigma}^{p}$ can be calculated with $\bm\theta_0$ replaced by $\widehat{\bm\theta}_{n}^{p}$ and the expectation approximated by the sample average. 
	The modified BIC in \eqref{BIC2} is adapt to the generalized BIC proposed by \cite{lv2014model_BIC}, where the addtional term of their generalized BIC is introduced to account for model misspecifications. 
	For more details of the derivation of the above BICs, please refer to Appendix A.3.1. 
	
	Let $\widehat{p}_{in} = \arg\min_{1\leq p\leq p_{\max}}\text{BIC}_i(p)$ for $i=1$ and 2, where $p_{\max}$ is a predetermined positive integer. Note that, when the sample size $n$ is sufficiently large, the additional terms $-(3p+1)\ln(2\pi)$ and $\ln(\det(\widehat{\Sigma}^{p}))$ can be ignored as they are $O(1)$. As a result, $\text{BIC}_1(p)$ and $\text{BIC}_2(p)$ are asymptotically equivalent in order selection. 
	Simulation results in Section \ref{simulation} indicate that the modified BIC in \eqref{BIC2} performs better than the original BIC in \eqref{BIC1} for moderate and small samples, although the two BICs have very similar performance for large samples. 
	Hence, for moderate and small samples, we suggest to use \eqref{BIC2}. 
	The following theorem verifies their selection consistency. 
	
	\begin{thm}\label{thm3BIC}
		Under the assumptions of Theorem \ref{thm2QMLE}, if $p_{max}\geq p_0$, then as $n\to \infty$,
		\[
		P(\widehat{p}_{1n} = p_0)\rightarrow 1 \quad \text{and} \quad P(\widehat{p}_{2n} = p_0)\rightarrow 1,
		\]
		where $p_0$ is the true order, and $p_{max}$ is a predetermined positive integer.
	\end{thm}
	
	\section{Testing for asymmetry} \label{test}
	
	This section studies the Wald, Lagrange multiplier (LM) and quasi-likelihood ratio (QLR) tests to detect the asymmetry (leverage) effect of news on volatilities.
	\subsection{Asymmetry Tests} \label{size}
	For model \eqref{model}, the asymmetry testing is of the form
	\begin{equation}\label{Hypotheses}
	H_0:\ \beta_{i0+} = \beta_{i0-}\ \text{for all } i \quad \text{against} \quad H_1: \beta_{i0+}\neq \beta_{i0-}\ \text{for some } i,
	\end{equation}
	where $i \in \left\lbrace 1,\cdots,p\right\rbrace$.
	Let $R = (0_{p\times(p+1)}, I_{p}, -I_{p})$ be the $p\times(3p+1)$ matrix, where $0_{m\times n}$ is the $m\times n$ zero matrix and $I_{p}$ is the $p\times p$ identity matrix. Then the null hypothesis can be represented as $H_0: R\bm\theta_{0} =\bm 0_{p}$, where $\bm\theta_{0}$ is the true parameter vector and $\bm 0_{p}$ is a $p$-dimensional zero vector.
	Hence, the Wald, LM and QLR test statistics are defined as
	\begin{equation}
	\begin{aligned}
	&W_n = n{\widehat{\bm \theta}_n^{\prime}R^{\prime}}(R\widehat{\Xi}R^{\prime})^{-1}R\widehat{\bm \theta}_{n},\\
	&L_n=\frac{1}{n}\frac{\partial L_n(\widetilde{\bm \theta}_{n})}{\partial \bm \theta^{\prime}}\widetilde{\Sigma}^{-1}R^{\prime}(R\widetilde{ \Xi}R^{\prime})^{-1}R\widetilde{\Sigma}^{-1}\frac{\partial L_n(\widetilde{\bm \theta}_{n})}{\partial \bm \theta},\\
	&Q_n=-2\left[ L_n(\widetilde{\bm \theta}_{n})-L_n(\widehat{\bm \theta}_n) \right],
	\end{aligned}
	\end{equation}
	respectively, where $\widetilde{\bm \theta}_n$ is the restricted QMLE under $H_0$ while $\widehat{\bm \theta}_n$ is the unrestricted QMLE,  $\widehat{\Xi}$ is the sample estimate of $\Xi$ with $\bm\theta_0$ estimated by $\widehat{\bm\theta}_{n}$ and the expectation replaced by sample average, while $\widetilde{\Xi}$ (or $\widetilde{\Sigma}$) is the sample estimate of $\Xi$ (or $\Sigma$) with $\bm\theta_0$ estimated by $\widetilde{\bm\theta}_{n}$ and the expectation replaced by sample average.
	
	Let $\chi^2_{\nu}$ be the chi-squared distribution with $\nu$ degrees of freedom.
	Define the $p\times p$ matrix $\Psi=\Delta^{-1/2}R\Xi R^\prime\Delta^{-1/2}$ with $\Delta=R\Sigma^{-1}R^\prime$. For $j=1,\ldots,p$, let $e_j$'s be the eigenvalues of $\Psi$, and $x_j$'s be the $i.i.d.$ random variables following the $\chi^2_1$ distribution. 
	The following theorem gives the limiting distributions of three test statistics under $H_0$. 
	
	\begin{thm}
		Suppose the assumptions of Theorem \ref{thm2QMLE} hold. Then, under $H_0$, as $n\rightarrow \infty$,
		\[(i)\, W_n\rightarrow_{\mathcal{L}}\chi^2_p;\quad (ii)\, L_n\rightarrow_{\mathcal{L}}\chi^2_p;\quad (iii)\, Q_n\rightarrow_{\mathcal{L}}Q;\]		
		where $Q=\sum_{j=1}^{p}e_jx_j$.
		\label{thm test asymmetry}
	\end{thm}
	
	Theorem \ref{thm test asymmetry} shows that the limiting null distribution of $Q_n$ is not the usual $\chi^2_p$ distribution but a distribution of the weighted sum of $i.i.d.$ $\chi^2_1$ random variables.
	This is because $\Sigma\neq\Omega$ in the absence of normality assumption on $\eta_{t}$; see also \cite{MaCurdy1981_quasiMLE_and_test_stat}.
	If $\eta_{t}$ is normally distributed, then $e_j=1$ for $j=1,\ldots, p$ and $Q$ reduces to a $\chi^2_p$ distribution, and hence $Q_n$ has the standard limiting null distribution as $W_n$ and $L_n$. 
	For general cases of $\eta_{t}$, we adopt the Pearson's three-moment central chi-square approach \citep{pearson1959note} to approximate $p$-values of the QLR test; see also \cite{imhof1961computing} and \cite{liu2009_liu_method}. The detailed procedure of Pearson's method is summarized in Remark \ref{p-value of QLR} below. 
	
	\begin{remark}\label{p-value of QLR} (Calculation of $p$-values for the QLR test) First, calculate $\mu_Q=c_1$, $\sigma_Q=\sqrt{2c_2}$ and $l=c_2^3/c_3^2$, where $c_k=\sum_{j=1}^{p}e_j^k$ for $k=1,2,3$. Then, the $p$-value of the QLR test is approximated by $P(\chi^2_l>(Q_n-\mu_Q)\sqrt{2l}/\sigma_Q+l)$, where $Q_n$ is the observed value of the QLR test statistic.
	\end{remark}

	\subsection{Power analysis} \label{power}
	
	We next discuss the efficiency of the proposed asymmetry tests through Pitman analysis.
	Note that $\bm\theta_0=(\bm\alpha_0^{\prime},\omega_0,\bm\beta_{0+}^{\prime},\bm\beta_{0-}^{\prime})^{\prime}$ with $\bm\beta_{0+}=\bm\beta_{0-}$ under $H_0$.
	Denote $\bm h=(\bm h_{\alpha}^{\prime},\bm h_{\beta}^{\prime})^\prime=(\bm h_{\alpha}^{\prime},h_w,\bm h_+^{\prime},\bm h_-^{\prime})^\prime \in \mathbb{R}^{p} \times \mathbb{R}_{+}^{2p+1}$, where $\bm h_{\alpha}=(h_{1},\ldots,h_{p})^{\prime}$, $\bm h_+=(h_{1+},\ldots,h_{p+})^{\prime}$, $\bm h_-=(h_{1-},\ldots,h_{p-})^{\prime}$ and $\bm h_{+}\neq\bm h_{-}$.
	Let $\bm\theta_n=\bm\theta_0+\bm h/\sqrt{n}$ such that $\bm\theta_n\in\Theta$ for sufficiently large $n$.
	Consider the local alternatives, that is, for each $n$, the observed time series $\{y_{p+1,n},\ldots,y_{n,n}\}$ are generated by
	\begin{equation}\label{LocalH1}
	H_{1n}:\quad y_{t,n} = \left(\bm\alpha_0+\dfrac{\bm h_{\alpha}}{\sqrt{n}}\right)^\prime \bm Y_{t-1,n} + \eta_t\left(\bm\beta_0+\dfrac{\bm h_{\beta}}{\sqrt{n}}\right)^\prime \bm X_{t-1,n},
	\end{equation}
	where the subscript $n$ is used to emphasize the dependence of $y_{t,n}$ on $n$, $\eta_t$ is defined as in the model \eqref{model}, $\bm Y_{t,n} = (y_{t,n},\ldots,y_{t-p+1,n})^{\prime}$, $\bm X_{t,n} =(1,\bm Y^{\prime}_{t+,n},-\bm Y^{\prime}_{t-,n})^{\prime}$ with $\bm Y_{t+,n} = (y_{t,n}^{+},\ldots,y_{t-p+1,n}^{+})^{\prime}$ and $\bm Y_{t-,n} = (y_{t,n}^{-},\ldots,y_{t-p+1,n}^{-})^{\prime}$. $\{y_{t,n}\}$ satisfies the condition below.
	
	\begin{assum}\label{assum_y_tH1}
		There exists a positive integer $n_0$ such that for $n\geq n_0$, $\left\lbrace y_{t,n}:t\in\mathcal{N}\right\rbrace $ is strictly stationary and geometrically ergodic with $E(|y_{t,n}|^\kappa)<\infty$ for some $\kappa>0$.
	\end{assum}
	
	Based on $\{y_{1,n},\ldots,y_{n,n}\}$, the QMLE under $H_{1n}$ can be defined as
	\begin{equation}
	\widehat{\bm \theta}_{n,h}=\arg \max_{\bm \theta \in \Theta} \sum_{t=p+1}^{n}\ell_{t,n}(\bm \theta),
	\label{QMLEH1}
	\end{equation}
	where
	\[\ell_{t,n}(\bm \theta)=-\ln \left(\bm \beta^{\prime} \bm X_{t-1,n}\right)-\dfrac{\left(y_{t,n}-\bm \alpha^{\prime} \bm Y_{t-1,n}\right)^{2}}{2\left(\bm{\beta}^{\prime} \bm{X}_{t-1,n}\right)^{2}}.\]
	Denote $\mathbb{P}_{n,h}$ as the law of $y_{t,n}$. The asymptotic distribution of $\widehat{\bm \theta}_{n,h}$ under sequences of local alternatives is given below.
	\begin{thm}\label{LAN}
		Suppose that Assumptions \ref{assum_eta_t}--\ref{assum_y_tH1} hold and $E(\eta_t^4)<\infty$, then, under $\mathbb{P}_{n,h}$, $\sqrt{n}(\widehat{\bm\theta}_{n,h}-\bm\theta_0)\rightarrow_{\mathcal{L}}N(\bm h,\Xi)$ as $n\to \infty$,
		where $\Xi$ is defined as in Theorem \ref{thm2QMLE}.
	\end{thm}
	Theorem \ref{LAN} verifies that model \eqref{LocalH1} is locally asymptotically normal \citep{Vaart2000_LAN} at $\bm\theta_0$.
	If $\eta_t$ follows a normal distribution, we can show Theorem \ref{LAN} by Le Cam's third lemma. However, when $\eta_{t}$ is not normal, the sequences $\mathbb{P}_{n,0}$ and $\mathbb{P}_{n,h}$ are not mutually contiguous; see Example 6.5 in \cite{Vaart2000_LAN} for more discussions. Therefore, we show Theorem \ref{LAN} in a direct way; see also \cite{Jiang2020}.
	
	Denote $D=\left(R \Xi R^\prime\right)^{1/2}\Delta^{-1}\left(R \Xi R^\prime\right)^{1/2}$, and define $\Gamma$ as an orthogonal matrix such that $\Gamma D \Gamma^\prime=\diag\{e_1^{\ast},\ldots,e_p^{\ast}\}$, where $e_j^{\ast}$'s are eigenvalues of $D$.  
	Denote $\bm v=\Gamma(R\Xi R^\prime)^{-1/2}R\bm h\in\mathbb{R}^p$, and let $v_j$ be its $j$-th component for $j=1, \ldots, p$.  
	Let $\chi^2_{\nu}(c)$ be the noncentral chi-squared distribution with degrees of freedom $\nu$ and noncentrality parameter $c$, and $\chi^2_{\nu,\tau}(c)$ be its $\tau$th quantile. 
	
	\begin{thm}\label{thm test asymmetryH1}
		Suppose the assumptions of Theorem \ref{LAN} hold and $R\bm\theta_{0} =\bm 0_{p}$. Then, under $\mathbb{P}_{n,h}$, as $n\to \infty$,
		\begin{equation*}
		\text{(i)}\ W_n\rightarrow \chi^2_p(\delta); \quad \text{(ii)}\ LM_n\rightarrow \chi^2_p(\delta);\quad\text{(iii)}\ Q_n\rightarrow \sum_{j=1}^{p}e_j^{\ast}x_{j,v^2_j};
		\end{equation*}
		where $\delta=\bm h^{\prime}R^{\prime}(R\Xi R^{\prime})^{-1}R \bm h$, and $x_{j,v^2_j}$'s are independent random variables following the $\chi^2_{1}(v^2_j)$ distribution for $j=1, \ldots, p$.
	\end{thm}
	Theorem \ref{thm test asymmetryH1} obtains the asymptotic distributions of three test statistics under the local alternatives,  
	which shows that the Wald and LM tests have the same local asymptotic powers.
	If $\eta_t$ is normally distributed, then $\sum_{j=1}^{p}e_j^{\ast}x_{j,v^2_j}$ reduces to a $\chi^2_p(\delta)$ distribution, and the QLR test is as efficient as the Wald and LM tests. 
	Moreover, note that $P(\chi^2_1(\delta) \geq \chi^2_{1,1-\tau}(\delta))=P(e_1^{\ast}\chi^2_{1}(v^2_1) \geq e_1^{\ast}\chi^2_{1,1-\tau}(v^2_1))$ holds for $p=1$, then it follows that the proposed three tests are equivalent in the local asymptotic power when $p=1$. 
	For general cases of $\eta_t$ with $p > 1$, it is difficult to compare the local asymptotic power of the QLR test with the other two tests. Alternatively, the simulation study in Section \ref{simulation_test} compares the local power of all three tests in finite samples, and it is found that three tests perform very similarly when the sample size is as large as 2000. 

	\section{Model checking} \label{checking}
	
	To check adequacy of the fitted asymmetric linear DAR model, we construct a mixed pormanteau test to detect misspecifications in the conditional mean and standard deviation jointly; see \cite{Wong2005mixed_portmanteau}. 
	In the literature, diagnostic checking the conditional mean and standard deviation, can be conducted by checking the significance of sample autocorrelation functions (ACFs) of residuals and absolute residuals, respectively.
	
	The ACFs of $\{\eta_t\}$ and $\{|\eta_t|\}$ at lag $k$ can be defined by $\rho_k=\cov(\eta_t, \eta_{t-k})/\var(\eta_t)$ and $\gamma_k=\cov(|\eta_t|, |\eta_{t-k}|)/\var(|\eta_t|)$, respectively.
	If the data generating process is correctly specified by model \eqref{model}, then $\{\eta_t\}$ and $\{|\eta_t|\}$ are $i.i.d.$ such that $\rho_k=0$ and $\gamma_k=0$ hold for any $k\geq 1$.
	For model \eqref{model} fitted by the QMLE, the corresponding residuals can be defined as $\widehat{\eta}_{t}=(y_t-\widehat{\bm\alpha}_n^\prime\bm Y_{t-1})/(\widehat{\bm\beta}_n^{\prime}\bm{X}_{t-1})$,
	and then the residual ACF and absolute residual ACF at lag $k$ can be calculated as
	\[
	\widehat{\rho}_{k}=\dfrac{\sum_{t=p+k+1}^{n}(\widehat{\eta}_{t}-\bar{\eta}_1)(\widehat{\eta}_{t-k}-\bar{\eta}_1)}{\sum_{t=p+1}^{n}(\widehat{\eta}_{t}-\bar{\eta}_1)^{2}} \text { and } \widehat{\gamma}_{k}=\dfrac{\sum_{t=p+k+1}^{n}(|\widehat{\eta}_{t}|-\bar{\eta}_2)(|\widehat{\eta}_{t-k}|-\bar{\eta}_2)}{\sum_{t=p+1}^{n}(|\widehat{\eta}_{t}|-\bar{\eta}_2)^{2}},
	\]
	respectively, where $\bar{\eta}_1=(n-p)^{-1}\sum_{t=p+1}^{n}\widehat{\eta}_t$ and $\bar{\eta}_2=(n-p)^{-1}\sum_{t=p+1}^{n}|\widehat{\eta}_t|$.
	Clearly, $\widehat{\rho}_{k}$ (or $\widehat{\gamma}_{k}$) is the sample version of ${\rho}_{k}$ (or ${\gamma}_{k}$).
	Accordingly, if the value of $\widehat{\rho}_{k}$ (or $\widehat{\gamma}_{k}$) deviates from zero significantly, it indicates that the conditional mean (or standard deviation) structure in model \eqref{model} is misspecified.
	
	Let $\widehat{\bm \rho}=(\widehat{\rho}_{1},\ldots,\widehat{\rho}_{M})^{\prime}$ and $\widehat{\bm \gamma}=(\widehat{\gamma}_{1},\ldots,\widehat{\gamma}_{M})^{\prime}$, where $M$ is a predetermined positive integer.
	Denote $\tau_1=E[\sgn(\eta_t)]$ and $\tau_2=E(|\eta_t|)$. Let $\xi_t=|\eta_t|-\tau_2$, then $E(\xi_t)=0$ and ${\sigma}_\xi^{2}=\var(\xi_t)=1-\tau_2^2$.
	Define the $M\times(3p+1)$ matrices $U_{\rho}=(\bm U_{\rho 1}^{\prime},\ldots,\bm U^{\prime}_{\rho M})^\prime$ and $U_{\gamma}=(\bm U^{\prime}_{\gamma 1},\ldots,\bm U^{\prime}_{\gamma M})^{\prime}$, where
	\[
	\bm U_{\rho k}=-\left(E\left(\dfrac{ \eta_{t-k} \bm {Y}_{t-1}^{\prime}}{\bm \beta_{0}^{\prime} \bm{X}_{t-1}}\right), \bm 0_{1 \times(2p+1)}\right) \text { and } \]
	\[\bm U_{\gamma k}=-\left(\tau_1E\left(\dfrac{\xi_{t-k} \bm{X}_{t-1}^{\prime}}{\bm \beta_{0}^{\prime} \mathbf{X}_{t-1}}\right), \tau_2E\left(\dfrac{\xi_{t-k} \bm{X}_{t-1}^{\prime}}{\bm \beta_{0}^{\prime} \mathbf{X}_{t-1}}\right)\right).
	\]
	Denote the $2M\times(2M+3p+1)$ matrix below
	\[
	V=\left(\begin{array}{ccc}
	I_{M} & 0 & U_{\rho} \\
	0 & I_{M} & U_{\gamma}/{\sigma}_\xi^{2}
	\end{array}\right).
	\]
	Let $\bm v_t=\left(\eta_t\eta_{t-1},\ldots,\eta_{t}\eta_{t-M},\xi_{t}\xi_{t-1}/{\sigma}_\xi^{2},\ldots,\xi_{t}\xi_{t-M}/{\sigma}_\xi^{2},
	\partial \ell_t(\bm \theta_{0})/\partial \bm\theta^\prime\Sigma^{-1}\right)^\prime$, and $G=E(\bm v_{t} \bm v_{t}^{\prime})$.
	\begin{thm}\label{thmACF}
		Suppose the assumptions of Theorem \ref{thm2QMLE} hold. If model \eqref{model} is correctly specified, then $\sqrt{n}(\widehat{\bm\rho}^{\prime},\widehat{\bm\gamma}^{\prime})^{\prime}\rightarrow_{\mathcal{L}} N\left(0, V G V^{\prime}\right)$ as $n\to \infty$.
	\end{thm}
	
	Theorem \ref{thmACF} can be used to check the significance of $\widehat{\rho}_{k}$ or $\widehat{\gamma}_{k}$ individually.
	We can construct consistent estimators of $V$ and $G$ using sample averages, which are denoted by $\widehat{G}$ and $\widehat{V}$, respectively. Then we can approximate the asymptotic distribution in Theorem \ref{thmACF}, and obtain confidence intervals for $\rho_{k}$ and $\gamma_{k}$.
	
	To check the first $M$ lags jointly, we construct a portmanteau test statistic below
	\begin{equation}
	Q(M)=n\left(\begin{array}{l}
	\widehat{\bm \rho} \\
	\widehat{\bm \gamma}
	\end{array}\right)^{\prime}\left(\widehat{V} \widehat{G} \widehat{V}^\prime\right)^{-1}\left(\begin{array}{l}
	\widehat{\bm \rho} \\
	\widehat{\bm \gamma}
	\end{array}\right). 
	\label{portmanteau test statistic}
	\end{equation}
	Theorem \ref{thmACF} and the continuous mapping theorem imply that, $Q(M)\to_{\mathcal{L}}\chi^2_{2M}$ as $n\to\infty$.
	Therefore, we reject the null hypothesis that $\rho_k$ and $\gamma_k$ ($1 \leq k\leq M$) are jointly insignificant at level $\tau$, if $Q(M)$ exceeds the $(1-\tau)$th quantile of $\chi^2_{2M}$ distribution.
	
	\section{Simulation experiments}\label{simulation}
	This section presents four simulation experiments to evaluate the finite-sample performance of the proposed QMLE, model selection method, three asymmetry tests and the mixed pormanteau test. 
	
	\subsection{Model estimation}\label{simulation_estimation}
	
	The first experiment aims to examine the finite-sample performance of the quasi-maximum likelihood estimator $\widehat{\bm \theta}_n$, for which the data generation process is
	\[
	y_t = 0.5 y_{t-1}+\eta_t(0.4+0.4y_{t-1}^+-0.6y_{t-1}^+),
	\]
	where $\{\eta_t\}$ are standard normal, or follow standardized Student $t_{5}$ distribution with unit variance, or standardized skewed $t$ distribution, denoted by $st_{5,-1.2}$, with unit variance and skew parameter $-1.2$ \citep{Jiang2020}. 
	The sample size is set to $n = 500, 1000$ or 2000, with 1000 replications for each sample size. 
	The projection newton method \citep{bertsekas1982projected} is employed for solving the optimization \eqref{optimization} in each replication. 
	Table \ref{tableQMLE} lists the biases, empirical standard deviations (ESDs) and asymptotic standard deviations (ASDs) of $\widehat{\bm \theta}_n$ for different innovation distributions and sample sizes. As the sample size increases, most of the biases, ESDs and ASDs become smaller, and the ESDs get closer to the corresponding ASDs. Moreover, when the distribution of $\eta_t$ gets more heavy-tailed or skewed, all ESDs and ASDs increase. This is as expected since either heavier tails or severer skewness of $\{\eta_t\}$ will lead to lower efficiency of the QMLE. 
	In addition, we also consider other parameter settings for the data generation process, and the simulation findings are similar as prevous. 
	
	\subsection{Model selection}\label{simulation_BIC}
	In the second experiment, we evaluate the performance of the proposed model selection method in Section \ref{BIC}, and compare the BIC$_1$ and its modified version BIC$_2$ in finite samples. 
	The data generating process is
	\[
	y_t = 0.3 y_{t-1}-0.2y_{t-2}+\eta_t(0.4+0.2y_{t-1}^+ +0.2y_{t-2}^+ -0.2y_{t-1}^- - 0.1y_{t-2}^-),
	\]
	where the innovations $\{\eta_t\}$ are defined as in the previous experiment. 
	Three sample sizes, $n=200, 500$ and 1000, are considered, and 1000 replications are generated for each sample size. 
	The BIC$_1$ in \eqref{BIC1} and BIC$_2$ in \eqref{BIC1} are employed to select the order $p$ with $p_{\max}=5$. 
	For $i = 1$ or 2, the cases of underfitting, correct selection and overfitting by BIC$_i$ correspond to $\widehat{p}_{in}$ being 1, 2 and greater than 2, respectively. 
	
	Table \ref{tableBIC} reports the percentages of underfitted, correctly selected and overfitted models by the two information criteria.
	The performance of both information criteria gets better when the sample size increases, while that becomes slightly worse as the distribution of $\eta_t$ gets more heavy-tailed or more skewed. 
	For the comparison between BIC$_1$ and BIC$_2$, it can be seen that the modified BIC (BIC$_2$) selects the correct model in most of the replications when the sample size is as small as $n=200$, while BIC$_1$ has comparable performance when the sample size is as large as $n=500$.  
	Overall, BIC$_2$ has better performance in model selection than that of BIC$_1$, especially for small and moderate samples. This indicates the necessity of the modified BIC in finite samples.

	\subsection{Asymmetry tests}\label{simulation_test}
	The third experiment examines the empirical size and power of the proposed asymmetry test statistics $W_n$, $L_n$ and $Q_n$. The data are generated from
	\[
	y_t = 0.4y_{t-1}+\eta_t[0.4+0.5y_{t-1}^+ - (0.5+k)y_{t-1}^-],
	\] 
	where $k=h/\sqrt{n}$ with $h\in \{-10,\ldots,-1,0,1,\ldots,10\}$ and $n$ being the sample size, and the innovations $\{\eta_t\}$ are defined as in the first experiment. 
	The null hypothesis of the asymmetry test is $H_0: k=0$, so that the case of $k=0$ corresponds to the size of the tests, the cases of $k\neq 0$ correspond to the local power. 
	Table \ref{tablesize} reports the empirical sizes of three tests at the significance level $5\%$ with $n=500, 1000$ and 2000. 
	From this table, we can see that, all tests have accurate sizes when the sample size is large. 
	In addition, $W_n$ and $Q_n$ are slightly oversized, especially when the sample size is small. 
	
	We next compare the local power of all three tests in finite samples at $5\%$ significance level. 
	Figure \ref{fig2} shows the empirical power of three tests for $n = 500$ and 2000. 
	We have the following findings. First, the local powers of $W_n$ and $Q_n$ are very similar, and they are slightly higher than that of $L_n$ when the sample size is small (i.e. $n = 500$), especially when $\eta_t$ is not normal or $|h|$ is not large. 
	Second, the local power of three tests is close to each other when the sample size is large (i.e. $n = 2000$), which is consistent to the theoretical comparison in Theorem \ref{thm test asymmetryH1} for $p=1$. 
	Finally, the local power for all three tests gets smaller as the innovations become more heavy-tailed or more skewed. 
	In addition, we also conduct simulation studies for the data generating process with $p>1$, the general findings are unchanged for the empirical size and power. 

	\subsection{Portmanteau test} \label{simulation_PT}
	
	In the fourth experiment, we study the proposed mixed portmanteau test $Q(M)$. The data are generated from
	\[
	y_t = 0.3y_{t-1}+c_1 y_{t-2}+\eta_t(0.4+0.3y_{t-1}^+ +{c_2}{y_{t-2}^+} -0.4{y_{t-1}^-}-c_2{y_{t-2}^-}),
	\]
	where the innovations $\{\eta_t\}$ are defined as in the first experiment. We fit an asymmetric linear DAR model with $p=1$ using the same method as in Section \ref{QMLE}, so that the case of $c_1=c_2=0$ corresponds to the size of the test, the case of $c_1 \neq 0$ corresponds to misspecifications in the conditional mean, and the case of $c_2>0$ corresponds to misspecifications in the conditional standard deviation. Two departure levels, 0.1 and 0.3, are considered for all $c_1$ and $c_2$.
	Table \ref{tablechecking} reports the rejection rates of $Q(6)$ at $5\%$ significance level based on 1000 replications, for sample size $n=500, 1000$ and $2000$. We have the following findings. 
	First, all sizes are close to the nominal level as the sample size $n$ increases, and most powers improve as $n$ or the departure level increases. 
	Second, $Q(6)$ is more powerful in detecting the misspecification in the conditional mean ($c_1\neq 0, c_2=0$) than that in the conditional standard deviation ($c_1=0$, $c_2>0$). 
	Finally, the performance of $Q(6)$ gets worse as the innovation distribution becomes more heavy-tailed or more skewed. This finding seems to be consistent with the result in the first experiment that, as the innovation distribution becomes more heavy-tailed or skewed, the estimation performance for all parameters tends to worsen.

	\section{An empirical example}\label{real_data}
	
	We illustrate the proposed inference tools using the weekly closing prices of S\&P500, denoted as $p_t$, span from January 1998 to December 2020, with 1200 observations in total. The data is downloaded from the website of Yahoo Finance (\textit{https://hk.finance.yahoo.com}). 
	Let $r_t=100\left( \ln p_t-\ln p_{t-1} \right)$ be the log returns in percentage, and denote $y_t=r_t-n^{-1}\sum_{t=1}^{n}r_t$ as the centered log returns in percentage. 
	The time plot of $\{y_t\}$ in Figure \ref{fig_real_data} suggests evident volatility clustering. 
	Table \ref{table_real_data_stat} lists summary statistics of $\{y_t\}$, where the sample skewness $-0.90$ indicates possible asymmetries in the volatility, and the sample kurtosis $7.17$ implies heavy-tailedness of $\{y_t\}$. 
	Moreover, the ACFs and partial ACFs of $\{y_t\}$ and $\{|y_t|\}$ are significant at the first few lags, which suggests that the autocorrelation coexists with the conditional heteroscedasticity in $\{y_t\}$. 
	The above findings motivate us to investigate $\{y_t\}$ by our proposed model and inference tools.
	
	Based on $p_{\max}=20$, the proposed BIC$_1$ and BIC$_2$ both select $p=4$. By the quasi-maximum likelihood estimation method in Section \ref{QMLE}, the fitted model is 
	\begin{align}\label{fitted-model-real-data}
	y_t=&-0.080_{0.033}y_{t-1}+0.034_{0.030}y_{t-2}+0.003_{0.032}y_{t-3}-0.014_{0.031}y_{t-4}+\widehat{\eta}_{t}\widehat{\sigma}_t \nonumber \\
	\widehat{\sigma}_t=&0.988_{0.096}+0.044_{0.050}y_{t-1}^{+}-0.415_{0.063}y_{t-1}^{-}+0.001_{0.047}y_{t-2}^{+}-0.248_{0.056}y_{t-2}^{-}\nonumber \\
	& +0.160_{0.050}y_{t-3}^{+}-0.286_{0.057}y_{t-3}^{-} +0.151_{0.047}y_{t-4}^+-0.189_{0.054}y_{t-4}^{-},
	\end{align}
	where the subscripts are the standard errors of the estimated coefficients. It can be seen that the coefficients of $y_{t-i}^+$ and $y_{t-i}^-$ are clearly different for $i=1,2,3$, which suggests that there may be asymmetric effects in the conditional volatility of $y_t$. 
	The Wald, LM and QLR tests in Section \ref{size} are conducted for model \eqref{fitted-model-real-data} and all their $p$-values are less than $0.001$, which corroborates the asymmetric effects in the volatility of $y_t$.
	To check the adequacy of the fitted model \eqref{fitted-model-real-data}, we perform the mixed portmanteau test $Q(M)$ in Section \ref{checking} for $M=6,12$ and $18$. The $p$-values of portmanteau tests are $0.41, 0.18$ and $0.27$, respectively, which suggests that the fitted model is adequate. In addition, as shown in Figure \ref{fig_real_data_portmanteau_test}, most of the residual ACFs  $\widehat{\rho}_k$ and $\widehat{\gamma}_k$ fall within their corresponding 95\% confidence bounds at the first 18 lags. 
	
	Since Value-at-Risk (VaR) is an important risk measure for financial assets, we use the fitted model to forecast the conditional quantile of $y_t$, i.e. the negative VaR.  
	To examine the forecasting performance, we conduct one-step-ahead predictions using a rolling forecasting procedure with a fixed moving window covering ten years' data points of size $522$. 
	Specifically, we fit an asymmetric linear DAR model of order four (ALDAR) for each moving window, and compute the forecast of the $\tau$th conditional quantile of $y_{t+1}$, given by $Q_{y_{t+1}}(\tau\mid \mathcal{F}_{t})=\widehat{\mu}_{t+1}+\widehat{\sigma}_{t+1}\widehat{b}_{\tau}$, where $\widehat{\mu}_{t+1}$ and $\widehat{\sigma}_{t+1}$ are the predicted conditional mean and standard deviation, respectively, and $\widehat{b}_\tau$ is the $\tau$th sample quantile of residuals $\{\widehat{\eta}_1,\ldots,\widehat{\eta}_t\}$. Then we move the window forward by one and repeat the above procedure until all data are used. Finally, we obtain 677 one-week-ahead negative VaRs for each $\tau$.  
	For illustration, the rolling forecasts at $\tau=5\%$ are displayed in Figure \ref{fig_real_data}, which indicates that the negative VaRs change accordingly to the volatility of the data. 
	
	To compare the forecasting performance of the proposed model with other counterparts, we also perform the rolling forecasting procedure using a linear DAR(4) model (LDAR) and an AR$(4)$ model with the threshold GARCH$(1,1)$ errors (AR-TGARCH). 
	Note that the AR-TGARCH model can depict the asymmetric effect in volatilities, while the LDAR model ignores the asymmetric effect. 
	For comparison, all these models are fitted by the QMLE, and their VaR forecasts are computed in the same way as for the ALDAR model. 
	To evaluate the forecasting performance of each model, we calculate the empirical coverage rate (ECR), and perform VaR backtests for the VaR forecasts at $\tau=1\%$, 5\%, 95\% and 99\%. 
	Specifically, ECR is calculated as the proportion of observations that fall below the corresponding conditional quantile forecast for the last 677 data points. Two VaR backtests, i.e. the likelihood ratio test for correct conditional coverage (CC) in \cite{christoffersen1998evaluating} and the dynamic quantile (DQ) test in \cite{Engle2004} are employed. Denote the hit by $H_t=I(y_t<Q_{y_t}(\tau\mid \mathcal{F}_{t-1}))$. The null hypothesis of CC test is that, conditional on $\mathcal{F}_{t-1}$, $\{H_t\}$ are $i.i.d.$ Bernoulli random variables with success probability being $\tau$. For the DQ test, following \cite{Engle2004}, we regress $H_t$ on regressors including a constant, four lagged hits $H_{t-i},i=1,2,3,4$, and the contemporaneous VaR forecast. The null hypothesis of DQ test is that all regression coefficients are zero and the intercept equals to the quantile level $\tau$. If the null hypothesis of each VaR backtest cannot be rejected, then it indicates that the VaR forecasts are satisfactory. 
	
	Table \ref{table_real_data_forecast} reports ECRs and $p$-values of two VaR backtests for the one-step-ahead forecasts by the fitted ALDAR, LDAR and AR-TGARCH models at the lower and upper 1\% and 5\% conditional quantiles, i.e. 1\% and 5\% VaRs for long and short positions. We use backtesting as the primary criterion, and the ECR as the secondary criterion. 
	In terms of backtests, none of the methods performs satisfactorily at $\tau=1\%$, and the proposed ALDAR model performs well at other three quantile levels with $p$-values not less than $0.5$. However, the LDAR model fails at all levels, and the AR-TGARCH model only performs adequately at $\tau=5\%$ with the $p$-values smaller than those of the ALDAR model. 
	For the ECRs, it can be seen that those of the ALDAR model are closest to the nominal quantile levels for upper quantiles. 
	The poor performance of the LDAR model is possibly because it ignores the asymmetric effect, while that of the AR-TGARCH model is perhaps because it is not robust to heavy-tailed data as its QMLE needs $E(y_t^4)<\infty$.   
	Therefore, we conclude that the proposed ALDAR model outperforms the other two competitors in forecasting VaRs for the S\&P500 Index.
	
	\section{Conclusion and discussion} \label{conclusion}
	
	This paper proposes the asymmetric linear double AR model which takes into account asymmetric effects for conditional heteroscedastic time series in the presence of a conditional mean structure. 
	The strict stationarity of the new model is derived, and inference tools, including a Gaussian QMLE for estimation and a mixed portmanteau test for diagnosis, are constructed without any moment condition on the data. 
	Based on the QMLE, a BIC and its modified version are proposed for order selection, and simulation results suggest that the modified BIC performs better in small and moderate samples. 
	The Wald, Lagrange multiplier and quasi-likelihood ratio test statistics are constructed to detect asymmetric effects, and it is shown that the Wald and Lagrange multiplier tests are asymptotically equivalent in size and power, while the asymptotics of the quasi-likelihood ratio test become non-standard.   
	The usefulness of the new model is confirmed by our empirical evidence, especially when the data are characterized by skewness and heavy-tailedness which are very common features for financial time series. 
	
	The study in this paper can be extended in several directions. 
	First, our model can be extended to allow for asymmetric effects in both the conditional mean and the standard deviation, then the proposed asymmetry tests could adapt to detect the asymmetry from the conditional location and scale separately or jointly. 
	Second, since financial time series can be heavy-tailed such that $E(\eta_t^4)=\infty$, it is also of interest to consider more robust estimation methods than Gaussian QMLE, for example, the quasi-maximum exponential likelihood estimation of \cite{Zhu_Ling2011}. 
	Third, the joint modeling of conditional mean and volatility in the presence of asymmetric effects for univariate case can be generalized to multivariate case. As a result, a vector asymmetric LDAR model is a natural extention and the related inference tools are worth to investigate.
	We leave these extensions for future research. 
	
	%
	
	\renewcommand{\thesection}{A}
	\setcounter{equation}{0} 
	\section*{Appendix: Technical proofs}
	
	This appendix includes technical details for Theorems \ref{thm1Stationarity}-\ref{thmACF}. 
	To show Theorems \ref{thm2QMLE} and \ref{LAN}, Lemmas \ref{boundedness}-\ref{LAN-lemma2} are introduced with proofs.  
	Throughout the appendix, for a vector $\bm x=(x_1, \ldots, x_p)^{\prime}$, the $n$-norm is defined as $\|\bm x\|_n=\left(\sum_{i=1}^p|x_i|^n\right)^{1/n}$; for a matrix or column vector $A$, we define $\|A\|=\sqrt{\text{tr}(AA')}$, where $\text{tr}(\cdot)$ denotes the trace of a square matrix. 
	
	\subsection{Proof of Theorem 1}
	
	\begin{proof}
		Denote $x^{+} = \max\left\lbrace 0,x\right\rbrace$ and $x^{-}= \min\left\lbrace 0,x\right\rbrace$. 
		Let $\bm Y_t=(y_t, \ldots, y_{t-p+1})^{\prime}$, $\bm Y^+_t=(y^+_t,\ldots,y^+_{t-p+1})^\prime$, $\bm Y^-_t=(y^-_t,\ldots,y^-_{t-p+1})^\prime$ and $\bm X_t=(1,{\bm Y^+_t}^\prime, -{\bm Y^-_t}^\prime)^{\prime}$, where $\{y_t\}$ are generated by model \eqref{model}. We begin by showing that $\{\bm Y_t\}$ is $\nu_p$-irreducible.
		
		Let $\mathcal{B}^p$ be the class of Borel sets of $\mathbb{R}^p$ and $\nu_p$ be the Lebesgue measure on $(\mathbb{R}^p, \mathcal{B}^p)$. Let $m: \mathbb{R}^p\rightarrow \mathbb{R}$  be the projection map onto the first coordinate, i.e. $m(\bm x)=x_1$ for $\bm x=(x_1, \ldots, x_p)^{\prime}$. 
		Then, $\{\bm Y_t\}$ is a homogeneous Markov chain on the state space $(\mathbb{R}^p, \mathcal{B}^p,\nu_p)$, with transition probability
		\[
		P(\bm x,A) =\int_{m(A)} \dfrac{1}{\bm x_{\sgn}^\prime \bm\beta} f\left(\dfrac{z-\bm x^{\prime}\bm\alpha}{\bm x_{\sgn}^\prime \bm\beta}\right)dz, \quad \bm x \in \mathbb{R}^p\; \text{and} \; A \in \mathcal{B}^p,
		\]
		where $\bm x_{\sgn}=(1,x_{1}^{+},\ldots,x_{p}^{+},-x_{1}^{-},\ldots,-x_{p}^{-})^{\prime}$, $\bm\alpha=(\alpha_1,...,\alpha_p)^{\prime}$, ${\bm\beta}=(\omega,\beta_{1+},\ldots,\beta_{p+},\\ \beta_{1-}, \ldots,\beta_{p-})^{\prime}$, and $f(\cdot)$ is the density function of $\eta_t$. We can further verify that the $p$-step transition probability of $\{\bm Y_t\}$ is 
		\begin{equation}\label{ptran}
		P^p(\bm x, A)=\int_{A}\prod_{i=1}^p\dfrac{1}{\bm X_{\sgn,i-1}^{\prime}\bm \beta} f\left(\dfrac{z_i-\bm X_{i-1}^{\prime}{\bm \alpha}}{\bm X_{\sgn,i-1}^{\prime}\bm \beta}\right)dz_1\ldots dz_p,
		\end{equation}
		where $\bm X_{i}=(z_i, \ldots, z_1, x_1, \ldots, x_{p-i})^{\prime}$ and $\bm X_{\sgn,i}$$=(1, z_i^+, \ldots, z_1^+, x_1^+,\ldots,x_{p-i}^+, -z_i^-, \ldots, -z_1^-,\\-x_1^-$$,\ldots,-x_{p-i}^-)^{\prime}$.
		Observe that, by Assumption \ref{assum1}, the transition density kernel in \eqref{ptran} is positive everywhere. As a result, $\{\bm Y_t\}$ is $\nu_p$-irreducible.
		
		We next prove that $\{\bm Y_{t}\}$ satisfies Tweedie's drift criterion \citep[Theorem 4]{Tweedie1983}, i.e., there exists a small set $G$ with $\nu_p(G)>0$ and a non-negative continuous function $g(\bm x)$ such that
		\begin{eqnarray}\label{drift criterion I}
		E\left\{g(\bm Y_{t})|\bm Y_{t-1}=\bm x\right\}\leq (1-\epsilon)g(\bm x), \quad \bm x \notin G,
		\end{eqnarray}
		\begin{eqnarray}\label{drift criterion II}
		E\left\{g(\bm Y_{t})|\bm Y_{t-1}=\bm x\right\}\leq M, \quad \bm x \in G,
		\end{eqnarray}
		for some constant $0<\epsilon<1$ and $0<M<\infty$. 
		We accomplish the proof in two parts, i.e. Case (i) for $0<\kappa\leq 1$ and Case (ii) for $\kappa \in \{2,3,4,\ldots\}$.  
		
		We first consider Case (i) for $0<\kappa\leq 1$. It can be verified that 
		\begin{align*}
		E(|y_{t+1}|^{\kappa}&\mid \bm Y_t=\bm x) \\
		&\leq \sum_{i=1}^p\left[E(|\alpha_i+\beta_{i+}\eta_{t+1}|^{\kappa})|x_i^+|^{\kappa}+E(|\alpha_i-\beta_{i-}\eta_{t+1}|^{\kappa})|x_i^-|^{\kappa} \right]+w^\kappa E(|\eta_{t+1}|^{\kappa})\\
		&\leq \sum_{i=1}^pa_{i}|x_i|^{\kappa} +w^\kappa E(|\eta_{t+1}|^{\kappa}),
		\end{align*}
		where $\bm x=(x_1,\ldots,x_p)^{\prime}$ and $a_{i}=\max\left\lbrace E(|\alpha_i+\beta_{i+}\eta_{t}|^{\kappa}), E(|\alpha_i-\beta_{i-}\eta_{t}|^{\kappa})\right\rbrace$ for $1\leq i\leq p$.
		Note that $\sum_{i=1}^pa_i<1$, and we can then find positive values $\{r_1,\ldots,r_{p-1}\}$ such that
		\begin{equation}\label{coe}
		a_p<r_{p-1}<1-\sum_{i=1}^{p-1}a_i \quad\text{and}\quad a_{i+1}+r_{i+1}<r_i<1-\sum_{k=1}^ia_k \;
		\text{for } 1\leq i \leq p-2.
		\end{equation}
		Consider the test function $g(\bm x)=1+|x_1|^\kappa+\sum_{i=1}^{p-1}r_i |x_{i+1}|^\kappa$, and we have that
		\begin{align*}\label{B.3}
		E\{g(\bm Y_{t+1})&\mid\bm Y_{t}=\bm x\} \nonumber \\
		&\leq1+\sum_{i=1}^{p}a_i|x_i|^{\kappa}+\sum_{i=1}^{p-1}r_i |x_i|^\kappa+\omega^{\kappa} E(|\eta_{t+1}|^{\kappa}) \nonumber \\
		&= 1+(a_1+r_1)|x_1|^\kappa+\sum_{i=2}^{p-1}\dfrac{a_i+r_i}{r_{i-1}}r_{i-1}|x_i|^\kappa +\dfrac{a_p}{r_{p-1}}r_{p-1}|x_p|^\kappa+\omega^{\kappa} E(|\eta_{t+1}|^{\kappa}) \nonumber \\
		&\leq \rho g(\bm x)+1-\rho+\omega^{\kappa} E(|\eta_{t+1}|^{\kappa}),
		\end{align*}
		where, from \eqref{coe}, 
		\begin{equation}\label{rho}
		\rho=\max\left\{a_1+r_1,\dfrac{a_2+r_2}{r_{1}},\cdots, \dfrac{a_{p-1}+r_{p-1}}{r_{p-2}},\dfrac{a_p}{r_{p-1}}\right\}<1.
		\end{equation}
		
		Denote $\epsilon=1-\rho-\{1-\rho+\omega^{\kappa} E(|\eta_{t+1}|^{\kappa})\}/g(\bm x)$, and $G=\{\bm x: \|\bm x\|\leq L\}$, where $L$ is a positive constant such that $g(\bm x)>1+\omega^{\kappa} E(|\eta_{t+1}|^{\kappa})/(1-\rho)$ as $\|\bm x\|>L$.
		We can verify that \eqref{drift criterion I} and \eqref{drift criterion II} hold, i.e. Tweedie's drift criterion holds. 
		Moreover, $\{\bm Y_{t}\}$ is a Feller chain since, for each bounded continuous function $g^*(\cdot)$, $E\{g^*(\bm Y_{t})|\bm Y_{t-1}=\bm x\}$ is continuous with respect to $\bm x$, and then $G$ is a small set.
		As a result, by Theorem 4(ii) in \cite{Tweedie1983} and Theorems 1 and 2 in \cite{Feigin_Tweedie1985}, $\{\bm Y_{t}\}$ is geometrically ergodic with a unique stationary distribution $\pi(\cdot)$, and
		\begin{equation*}
		\int_{\mathbb{R}^p}g(\bm x)\pi(d\bm x)=1+\left(1+\sum_{i=1}^{p-1}r_i\right)E(|y_t|^{\kappa})<\infty,
		\end{equation*}
		which implies that $E(|y_t|^{\kappa})<\infty$.
		This accomplishes the first part.
		
		Next, we consider Case (ii) for $\kappa \in \{2,3,4,\ldots\}$. 
		Note that $y_t=\sum_{i=1}^{p}[(\alpha_i+\beta_{i+}\eta_{t})x_i^++(\alpha_i-\beta_{i-}\eta_{t})x_i^-]+\eta_{t}\omega$.  
		This together with the multinomial theorem, implies that
		\begin{align*}
		E(|y_{t+1}|^{\kappa}&\mid \bm Y_t=\bm x) \\
		\leq &E\left[\left(\sum_{i=1}^{p}[|\alpha_i+\beta_{i+}\eta_{t+1}||x_i^+|+|\alpha_i-\beta_{i-}\eta_{t+1}||x_i^-|]+|\eta_{t+1}|\omega\right)^\kappa\right]\\
		\leq &E\left[\left( \sum_{i=1}^{p}\delta_i|x_i|+|\eta_{t+1}|\omega \right)^{\kappa} \right]\\
		= &E\left( \sum_{\kappa_1+\cdots+\kappa_p=\kappa}\dfrac{\kappa!}{\kappa_1!\cdots\kappa_p!}\delta^{\kappa_1}_1\cdots\delta_p^{\kappa_p}|x_1|^{\kappa_1}\cdots|x_p|^{\kappa_p}+\right.\\
		&\left.\sum_{\kappa_1+\cdots+\kappa_p<\kappa}\dfrac{\kappa!}{\kappa_1!\cdots\kappa_p!(\kappa-\sum_{j=1}^p\kappa_j)!}\delta^{\kappa_1}_1\cdots\delta_p^{\kappa_p}|x_1|^{\kappa_1}\cdots|x_p|^{\kappa_p}|\eta_{t+1}\omega|^{\kappa-\sum_{j=1}^p\kappa_j}\right),
		\end{align*}
		where $\delta_i=\max\left\lbrace|\alpha_i+\beta_{i+}\eta_{t+1}|,|\alpha_i-\beta_{i-}\eta_{t+1}|\right\rbrace$. 
		For positive integers $n_1,\cdots,n_p$, we have
		\begin{equation}\label{fact1}
		|x_1|^{n_1}\cdots|x_p|^{n_p}\leq \dfrac{\sum_{i=1}^{p}n_i|x_i|^{n_1+\cdots+n_p}}{n_1+\cdots+n_p},
		\end{equation}
		and if $n>n_1+\cdots+n_p$, 
		\begin{equation}\label{fact2}
		\dfrac{|x_1|^{n_1}\cdots|x_p|^{n_p}}{||\bm x||_n^n}=o(1)\to 0 \;\;\text{as} \;\; \|\bm x\|_n\to\infty. 
		\end{equation}
		Then by \eqref{fact1} and \eqref{fact2}, it can be verified that
		\begin{align}\label{conditional y}
		E(|y_{t+1}|^{\kappa}&\mid \bm Y_t=\bm x) \nonumber\\ 
		\leq &E\left( \sum_{\kappa_1+\cdots+\kappa_p=\kappa}\dfrac{\kappa!}{\kappa_1!\cdots\kappa_p!}\delta^{\kappa_1}_1\cdots\delta_p^{\kappa_p}|x_1|^{\kappa_1}\cdots|x_p|^{\kappa_p}+o(||\bm x||_\kappa^\kappa)\right) \nonumber\\
		\leq  &E\left( \sum_{\kappa_1+\cdots+\kappa_p=\kappa}\dfrac{\kappa!}{\kappa_1!\cdots\kappa_p!}\delta^{\kappa_1}_1\cdots\delta_p^{\kappa_p}\dfrac{\sum_{i=1}^{p}\kappa_i|x_i|^\kappa}{\kappa}+o(||\bm x||_\kappa^\kappa)\right) \nonumber\\
		=&E\left( \sum_{i=1}^{p}\delta_i|x_i|^\kappa\sum_{\kappa_1+\cdots+\kappa_p=\kappa}\dfrac{(\kappa-1)!}{\kappa_1!\cdots(\kappa_i-1)!\cdots\kappa_p!}\delta^{\kappa_1}_1\cdots\delta_i^{\kappa_i-1}\cdots\delta_p^{\kappa_p} \right)+o(||\bm x||_\kappa^\kappa) \nonumber\\
		=&E\left( \sum_{i=1}^{p}\delta_i\left(\sum_{j=1}^{p}\delta_j\right)^{\kappa-1}|x_i|^\kappa\right)+o(||\bm x||_\kappa^\kappa) \nonumber\\
		=&\sum_{i=1}^{p}a_{i}|x_i|^\kappa+o(||\bm x||_\kappa^\kappa),
		\end{align}
		where $a_{i}=E[\delta_i(\sum_{j=1}^{p}\delta_j)^{\kappa-1}]$ for $1\leq i\leq p$. 
		Note that for $\kappa \in \{2,3,4,\ldots\}$, by the assumption of Case (ii), we have
		$$\sum_{i=1}^pa_{i}=E\left( \sum_{i=1}^{p}\delta_i\left(\sum_{j=1}^{p}\delta_j\right)^{\kappa-1} \right)=E\left( \sum_{i=1}^{p}\max\left\lbrace|\alpha_i+\beta_{i+}\eta_{t+1}|,|\alpha_i-\beta_{i-}\eta_{t+1}|\right\rbrace \right)^\kappa<1.$$ 
		As a result, we can find positive values $\{r_1,\ldots,r_{p-1}\}$ such that \eqref{coe} holds.
		
		Consider the test function $g(\bm x)=1+|x_1|^\kappa+\sum_{i=1}^{p-1}r_i |x_{i+1}|^\kappa$ as for Case (i). Define $\rho$ as in \eqref{rho}. Note that $g(\bm x)=O(1+||\bm x||_\kappa^\kappa)$, this together with \eqref{conditional y}, implies that
		\begin{align*}\label{B.4}
		E\{g(\bm Y_{t+1})&\mid\bm Y_{t}=\bm x\} \nonumber \\
		&\leq1+\sum_{i=1}^{p}a_i|x_i|^{\kappa}+\sum_{i=1}^{p-1}r_i |x_i|^\kappa+o(||\bm x||_\kappa^\kappa) \nonumber \\
		&= 1+(a_1+r_1)|x_1|^\kappa+\sum_{i=2}^{p-1}\dfrac{a_i+r_i}{r_{i-1}}r_{i-1}|x_i|^\kappa +\dfrac{a_p}{r_{p-1}}r_{p-1}|x_p|^\kappa+o(||\bm x||_\kappa^\kappa) \nonumber \\
		&\leq \rho g(\bm x)+1-\rho+o(||\bm x||_\kappa^\kappa)=(\rho+o(1))g(\bm x), 
		\end{align*}
		where $o(1)\to 0$ as $||\bm x||_\kappa\rightarrow\infty$. 
		For any fixed $\epsilon>0$, choose $L>0$ large enough, such that $\rho+o(1)<1-\epsilon<1$, as $||\bm x||_\kappa>L$. 
		Let $G=\{\bm x: \|\bm x\|\leq L\}$, then $G$ is a bounded set with $\mu_p(G)>0$. 
		It can be shown that \eqref{drift criterion I} and \eqref{drift criterion II} hold, i.e. Tweedie's drift criterion is verified. Similar to the proof of Case (i), we can show that, there exists a strictly stationary solution $\{ y_t \}$ to model \eqref{model}, and this solution is unique and geometrically ergodic with $E\left(|y_t|^\kappa \right)<\infty$. This accomplishes the second part. 
		
	\end{proof}
	
	\subsection{Proof of Theorem \ref{thm2QMLE}}
	
	To show Theorem \ref{thm2QMLE}, we introduce the following lemmas. 
	\begin{lemma}\label{boundedness}
		If Assumptions \ref{assum_y_t} and \ref{assumption_compact} hold, then it holds that
		\[
		\text{ (i) } E\sup _{\bm \theta \in \Theta}\left|\ell_{t}(\bm \theta)\right|<\infty; \quad \text{ (ii) } E\sup_{\bm \theta \in \Theta} \left\| \dfrac{\partial \ell_{t}(\bm \theta)}{\partial \bm \theta } \right\|<\infty; \quad \text{ (iii) } E\sup _{\bm \theta \in \Theta}\left\|\dfrac{\partial^{2} \ell_{t}(\bm \theta)}{\partial \bm \theta \partial \bm \theta^{\prime}}\right\|<\infty.
		\]
	\end{lemma}
	\begin{proof}
		Recall that $\bm \theta = (\bm \alpha^\prime,\bm \beta^\prime)^{'}$ and $\ell_{t}(\bm \theta)=-\ln \left(\bm \beta^{\prime} \bm X_{t-1}\right)-0.5\left(y_{t}-\bm \alpha^{\prime} \bm Y_{t-1}\right)^{2}/\left(\bm{\beta}^{\prime} \bm{X}_{t-1}\right)^{2}$, where $\bm Y_{t} = (y_t,\ldots,y_{t-p+1})^{\prime}$ and $\bm X_t =(1,\bm Y^{\prime}_{t+},-\bm Y^{\prime}_{t-})^{\prime}$ with $\bm Y_{t+} = (y_t^{+},\ldots,y_{t-p+1}^{+})^{\prime}$ and $\bm Y_{t-} = (y_t^{-},\ldots,y_{t-p+1}^{-})^{\prime}$. 
		It can be derived that
		\begin{align*}
		&\dfrac{\partial \ell_{t}(\bm\theta)}{\partial \bm\alpha}=\dfrac{\bm Y_{t-1}(y_t-\bm \alpha^\prime \bm Y_{t-1})}{(\bm \beta^\prime \bm X_{t-1})^2}, \quad 
		\dfrac{\partial \ell_{t}(\bm \theta)}{\bm \beta}=-\dfrac{\bm X_{t-1}}{\bm \beta^\prime \bm X_{t-1}}\left[ 1-\dfrac{(y_t-\bm \alpha^\prime\bm X_{t-1})^2}{(\bm \beta^\prime \bm X_{t-1})^2} \right],\\
		&\dfrac{\partial^2 \ell_{t}(\bm\theta)}{\partial \bm\alpha \partial \bm\alpha^\prime}=-\dfrac{\bm Y_{t-1}\bm Y_{t-1}^\prime}{(\bm\beta^\prime \bm X_{t-1})^2}, \quad 
		\dfrac{\partial^2 \ell_{t}(\bm\theta)}{\partial \bm\beta \partial \bm\beta^\prime}=\dfrac{\bm X_{t-1}\bm X_{t-1}^\prime}{(\bm\beta^\prime \bm X_{t-1})^2}\left[ 1-\dfrac{3(y_t-\bm\alpha^\prime \bm Y_{t-1})^2}{(\bm\beta^\prime \bm X_{t-1})^2} \right], \;\text{and}\\
		&\dfrac{\partial^2 \ell_{t}(\bm\theta)}{\partial \bm\alpha \partial \bm\beta^\prime}=-\dfrac{2\bm Y_{t-1}\bm X_{t-1}^\prime(y_t-\bm\alpha^\prime \bm Y_{t-1})}{(\bm\beta^\prime \bm X_{t-1})^3}.
		\end{align*}
		We first show (i). By Assumption \ref{assum_y_t}, there exists a constant $\kappa >0$ such that $E(|y_t|^\kappa) < \infty$. Denote $\overline{\omega}^\star = \max\left\lbrace 1,\overline \omega\right\rbrace $ and $c=I(0<\kappa\leq 1)+(p+1)I(\kappa>1)$, i.e. $c=1$ if $\kappa\leq 1$ and $c=p+1$ if $\kappa>1$. By the $c_r$ inequality and Jensen's inequality, we have
		\begin{align*}
		E \ln \left(\bar{\omega}^{\star}+\bar{\beta}\sum_{i=1}^{p} |y_{t-i}|\right) &=\frac{1}{\kappa} E \ln \left(\bar{\omega}^{\star}+\bar{\beta}\sum_{i=1}^{p}  |y_{t-i}|\right)^{\kappa} \\ 
		&\leq \frac{1}{\kappa} E \ln \left(c^{\kappa-1}\bar{\omega}^{\star \kappa}+c^{\kappa-1}\bar{\beta}^\kappa\sum_{i=1}^{p} |y_{t-i}|^\kappa\right) \\ 
		&\leq \frac{1}{\kappa} \ln \left(c^{\kappa-1}\bar{\omega}^{\star \kappa}+c^{\kappa-1}\bar{\beta}^\kappa\sum_{i=1}^{p} E|y_{t-i}|^\kappa\right) <\infty.
		\end{align*}
		This together with Assumption \ref{assumption_compact} and $|y_{t}|=y^+_{t}- y^-_{t}$, implies that
		\begin{align}\label{finite 1}
		&E \sup _{\bm\theta \in \Theta}\left|\ln \left(\omega+\sum_{i=1}^{p} (\beta_{i+} y^+_{t-i}-\beta_{i-} y^-_{t-i})\right)\right|  \nonumber \\
		\leq &E \sup _{\bm\theta \in \Theta}\left[I\left(\omega+\sum_{i=1}^{p} \bar{\beta} (y^+_{t-i}- y^-_{t-i}) \geq 1\right) \ln \left(\omega+\sum_{i=1}^{p} \bar{\beta} (y^+_{t-i}- y^-_{t-i})\right)\right]  \nonumber\\
		&+E \sup _{\bm\theta \in \Theta}\left[-I\left(\omega+\sum_{i=1}^{p}\bar{\beta} (y^+_{t-i}- y^-_{t-i}) \leq 1\right) \ln \left(\omega+\sum_{i=1}^{p} \bar{\beta}(y^+_{t-i}- y^-_{t-i})\right)\right]  \nonumber\\
		\leq & E \ln \left(\bar{\omega}^{\star}+\bar{\beta}\sum_{i=1}^{p}  |y_{t-i}|\right)-I\{\underline{\omega}<1\} \ln \underline{\omega}<  \infty.		
		\end{align}
		Note that $\eta_t$ is independent of $\mathcal{F}_{t-1}$, $E(\eta_t)=0, E(\eta_t^2)=1$ and  
		$y_{t}-\sum_{i=1}^{p} \alpha_{i} y_{t-i}=\sum_{i=1}^{p}(\alpha_{i0}-\alpha_{i})y_{t-i}+\eta_t\left(\omega_0+\sum_{i=1}^{p}\left(\beta_{i0+} y_{t-i}^{+}-\beta_{i0-} y_{t-i}^{-}\right)\right)$, 
		then by Assumption \ref{assumption_compact} and $c_r$ inequality, it can be verified that
		\begin{align}\label{finite 2}
		&E\sup _{\bm \theta \in \Theta}\left[\frac{\left(y_{t}-\sum_{i=1}^{p} \alpha_{i} y_{t-i}\right)^{2}}{(\omega+\sum_{i=1}^{p} \beta_{i+} y^+_{t-i}-\beta_{i-} y^-_{t-i})^2}\right] \nonumber\\
		\leq &E\sup _{\bm \theta \in \Theta}\left[\left(\sum_{i=1}^{p}\dfrac{\left(\alpha_{i0}-\alpha_{i}\right) y_{t-i}}{\underline{\omega}+\underline{\beta}\sum_{i=1}^{p} |y_{t-i}|}\right)^{2}\right]
		+ E \left[\left(\dfrac{\overline{\omega}+\bar{\beta}\sum_{i=1}^{p} |y_{t-i}|}{\underline{\omega}+\underline{\beta}\sum_{i=1}^{p} |y_{t-i}|}\right)^2\right]  \nonumber\\
		\leq & E\sup _{\bm \theta \in \Theta}\left[\left(\sum_{i=1}^{p}\dfrac{\left(\alpha_{i0}-\alpha_{i}\right) y_{t-i}}{\underline{\beta} |y_{t-i}|}\right)^{2}\right] \nonumber\\
		&+
		2E\left[\left(\dfrac{\overline{\omega}}{\underline{\omega}+\underline{\beta}\sum_{i=1}^{p} |y_{t-i}|}\right)^2\right]
		+2E\left[ \left(\sum_{i=1}^{p}\dfrac{\bar{\beta}|y_{t-i}|}{\underline{\omega}+\underline{\beta}\sum_{i=1}^{p} |y_{t-i}|}\right)^2\right] \nonumber\\
		\leq & \dfrac{p}{\underline{\beta}^2}\sum_{i=1}^{p}\sup _{\bm \theta \in \Theta}(\alpha_{i0}-\alpha_{i})^2 +\dfrac{2\overline{\omega}^2}{\underline{\omega}^2}+\dfrac{2p^2\bar{\beta}^2}{\underline{\beta}^2}<\infty. 
		\end{align}
		By \eqref{finite 1}, \eqref{finite 2} and the triangle inequality, we have
		\[E\sup_{\bm \theta \in \Theta}\left|\ell_{t}(\bm \theta)\right|
		\leq E\sup_{\bm\theta \in \Theta}\left|\ln\left(\bm \beta^{\prime} \bm X_{t-1}\right)\right|
		+ \dfrac{1}{2}E\sup_{\bm \theta \in \Theta}\dfrac{\left(y_{t}-\bm \alpha^{\prime} \bm Y_{t-1}\right)^{2}}{\left(\bm{\beta}^{\prime} \bm{X}_{t-1}\right)^{2}}<\infty.\]
		Thus, (i) is verified. 
		Similarly, we can show that (ii) and (iii) hold. 
	\end{proof}

	\begin{lemma}\label{three o_p}
		If Assumptions \ref{assum_y_t} and \ref{assumption_compact} hold, then 
		\begin{equation*}
		\begin{aligned}
		&\text { (i) } \sup_{\bm \theta \in \Theta}\left|\dfrac{1}{n} \sum_{t=p+1}^n \ell_t(\bm \theta)-E \left[\ell_{t}(\bm \theta)\right]\right|=o_{p}(1); \\ 
		&\text{ (ii) } \sup_{\bm \theta \in \Theta}\left\|\dfrac{1}{n} \sum_{t=p+1}^n \dfrac{\partial \ell_t(\bm \theta)}{\partial \bm \theta}-E\left[\dfrac{\partial \ell_{t}(\bm \theta)}{\partial \bm \theta}\right]\right\|=o_{p}(1); \\ 
		&\text { (iii) } \sup _{\bm \theta \in \Theta}\left\|\dfrac{1}{n} \sum_{t=p+1}^{n}\frac{\partial^{2} \ell_{t}(\bm \theta)}{\partial \bm \theta \partial \bm \theta^{\prime}}-E\left[\frac{\partial^{2} \ell_{t}(\bm \theta)}{\partial \bm \theta \partial \bm \theta^{\prime}}\right]\right\|=o_{p}(1).\\
		\end{aligned}
		\end{equation*}
	\end{lemma}
	
	\begin{proof}
		These follow from Lemma \ref{boundedness} and Theorem 3.1 in \cite{Ling_McAleer2003}.
	\end{proof}
	
	\begin{lemma}\label{maxmum at true parm}
		If Assumptions \ref{assum_y_t} and \ref{assumption_compact} hold, then $E\ell_t(\bm \theta)$ has a unique maximum at $\bm \theta_0$. 
	\end{lemma}
	\begin{proof}
		We first prove that	
		\begin{equation}
		\bm c_1 = \bm 0\quad \text{if} \quad \bm c^{\prime}_1 \bm Y_{t} = 0\ a.s. \quad \text{and}\quad  \bm c_2 = \bm 0\quad \text{if} \quad \bm c^{\prime}_2 \bm X_{t} = 0\ a.s.,
		\label{constant Vector}
		\end{equation}
		where $\bm c_1$ and $\bm c_2$ are $p\times 1$ and $(2p+1)\times 1$ constant vectors, respectively.
		If $\bm c^{\prime}_1 \bm Y_{t} = 0\ a.s.$ and $\bm c_1=(c_{11},\cdots,c_{1p})^{\prime} \neq \bm 0$, without loss of generality, we can assume $c_{11}=1$, thus $y_t = -\sum_{i=2}^{p}c_{1i}y_{t-i+1}\ a.s.$. Recall that $\eta_{t}=(y_t-\bm \alpha_{0}^\prime \bm Y_{t-1})/( \bm \beta_{0}^\prime \bm X_{t-1})$ and $\eta_t$ is independent of $\mathcal{F}_{t-1}$, we have
		\begin{equation}\label{contradiction}
		E(\eta^2_t)=E(\eta_t)E\left(\dfrac{-\sum_{i=2}^{p}c_{1i}y_{t-i+1}-\bm \alpha_{0}^{\prime}\bm Y_{t-1}}{\bm \beta^{\prime}_{0}\bm X_{t-1}}\right)= 0,
		\end{equation}
		which is a contradiction with $E(\eta^2_t)=1$, thus $\bm c_1=\bm 0$. 
		
		Denote $\bm c_2=(d_0,d_{1+},\ldots,d_{p+},d_{1-},\ldots,d_{p-})^\prime$, if $\bm c^{\prime}_2 \bm X_{t} = 0\ a.s.$ and $\bm c_2 \neq 0$, without loss of generality, we assume $d_{1+} =1$, then $y^{+}_t = d_{1-}y^-_t-d_0-\sum_{i=2}^{p}(d_{i+}y^+_{t-i+1}-d_{i-}y^-_{t-i+1})$. On the one hand, if $d_{1-}=-d_{1+}$, then $y_t=y_t^++y_t^-=-d_0-\sum_{i=2}^{p}(d_{i+}y^+_{t-i+1}-d_{i-}y^-_{t-i+1})$, similar to \eqref{contradiction}, we can find a contradiction with $E(\eta^2_t)=1$. On the other hand, if $d_{1-}\neq -d_{1+}$, i.e. $d_{1-}\neq -1$, note that $y_t=y_t^-+y_t^+$, then we have
		\begin{align}\label{y_t^- and eta_t}
		(1+d_{1-})y_t^- =&d_{0}+\alpha_{10}y_{t-1}+\sum_{i=2}^p\left[(\alpha_{i0}+d_{i+})y_{t-i+1}^++(\alpha_{i0}-d_{i-})y_{t-i+1}^-\right] \nonumber\\
		& +\eta_t\left(\omega_0+\sum_{i=1}^{p}(\beta_{i0+}y_{t-i}^+-\beta_{i0-}y_{t-i}^-)\right).
		\end{align}
		Without loss of generality, we assume $1+d_{1-}>0$. Denote 
		\[M=-\dfrac{d_{0}+\alpha_{10}y_{t-1}+\sum_{i=2}^p\left[(\alpha_{i0}+d_{i+})y_{t-i+1}^++(\alpha_{i0}-d_{i-})y_{t-i+1}^-\right]}{\omega_0+\sum_{i=1}^{p}(\beta_{i0+}y_{t-i}^+-\beta_{i0-}y_{t-i}^-)}.\]
		Since the density function of $\eta_{t}$ is positive everywhere on $\mathbb{R}$ by Assumption \ref{assum_eta_t}, by a simple transformation on \eqref{y_t^- and eta_t}, we can obtain that $P(y_{t}^->0\mid \mathcal{F}_{t-1})=P(\eta_{t}>M\mid \mathcal{F}_{t-1})>0$, which contradicts the condition that  $P(y_t^->0)=0$. Thus, $\bm c_2 = \bm 0$ if $\bm c^{\prime}_2 \bm X_{t} = \bm 0\ a.s.$. Hence, \eqref{constant Vector} is verified. 
		
		As for \eqref{constant Vector}, we can show that	
		\begin{align}\label{E_ell}
		E \ell_{t}(\bm \theta)=&-E\left[\ln \left(\bm \beta^{\prime} \bm X_{t-1}\right)+\frac{\left(y_{t}-\bm \alpha^{\prime}\bm Y_{t-1}\right)^{2}}{2(\bm \beta^{\prime} \bm X_{t-1})^2}\right]\notag \\
		=&-E\left[\ln \left(\bm \beta^{\prime} \bm X_{t-1}\right)+\dfrac{1}{2}\left(\dfrac{\bm \beta_{0}^{\prime} \bm X_{t-1}}{\bm \beta^{\prime} \bm X_{t-1}}\right)^2\right]-\dfrac{1}{2} E\left[\dfrac{\left(\bm \alpha-\bm \alpha_{0}\right)^{\prime} \bm Y_{t-1}}{\bm \beta^{\prime} \bm X_{t-1}}\right]^2.
		\end{align}
		The second term in \eqref{E_ell} reaches its maximum at zero, and this happens if and only if $(\bm \alpha-\bm \alpha_0)^{\prime}\bm Y_{t-1} = 0\ a.s.$, which holds if and only if $\bm \alpha = \bm \alpha_{0}$ by \eqref{constant Vector}. 
		For the first term in \eqref{E_ell}, denote $f(x) = -\ln(x) - 0.5a^2/x^2,$ where $x = \bm \beta^{\prime} \bm X_{t-1}$ and $a = \bm \beta_0^{\prime} \bm X_{t-1}$. We can prove that $f(x)$ reaches its maximum at $x = a$, i.e. $\bm \beta^{\prime} \bm X_{t-1} = \bm \beta_{0}^{\prime} \bm X_{t-1}$, which holds if and only if $\bm \beta = \bm \beta_{0}$ by \eqref{constant Vector}. 	
		Therefore, $E\ell_t(\bm \theta)$ is uniquely maximized at $\bm \theta_0$.
	\end{proof}
	\begin{lemma} \label{MCLT}
		Suppose Assumptions \ref{assum_y_t} and \ref{assumption_compact} hold. If $E(\eta_t^4)<\infty$ and $D=\left(\begin{matrix}
		1 & \kappa_1\\
		\kappa_1& \kappa_2
		\end{matrix}
		\right)$ is positive definite, then
		\item[(i)] $\Omega$ and $\Sigma$ are finite and positive definite;
		\item[(ii)] $ n^{-1/2} \sum_{t=p+1}^{n} \partial \ell_{t}\left(\bm \theta_{0}\right)/\partial \bm \theta \longrightarrow_{\mathcal{L}} N(0,\Omega)$.
	\end{lemma}
	\begin{proof}
		We first show (i). Recall that $\kappa_1=E\eta_t^3$, $\kappa_2=E\eta_t^4-1$, 
		\[
		\Omega= E\left(\begin{array}{cc}
		{\dfrac{\bm Y_{t-1} \bm Y_{t-1}^{\prime}}{\left(\bm \beta_{0}^{\prime} \bm X_{t-1}\right)^{2}}} & {\dfrac{\kappa_1\bm Y_{t-1} \bm X_{t-1}^{\prime} }{\left(\bm \beta_{0}^{\prime} \bm X_{t-1}\right)^{2}}} \\
		{\dfrac{\kappa_1\bm X_{t-1} \bm Y_{t-1}^{\prime}}{\left(\bm \beta_{0}^{\prime} \bm X_{t-1}\right)^{2}}} & {\dfrac{\kappa_2\bm X_{t-1} \bm X_{t-1}^\prime}{\left(\bm \beta_{0}^{\prime} \bm X_{t-1}\right)^{2}}}
		\end{array}\right)\; \text{and} \; 
		\Sigma = \diag\left\{E\left[\frac{\bm Y_{t-1} \bm Y_{t-1}^{\prime}}{(\bm \beta_{0}^{\prime} \bm X_{t-1})^2}\right], E\left[\frac{2 \bm X_{t-1} \bm X_{t-1}^{\prime}}{\left(\bm{\beta}_{0}^{\prime} \bm{X}_{t-1} \right)^{2}}\right]\right\}.
		\]
		By Assumptions \ref{assum_y_t} and \ref{assumption_compact}, for some constant $C$, we have
		\[
		E\left\|\dfrac{\bm Y_{t-1} \bm Y_{t-1}^\prime}{(\bm \beta_0^\prime \bm X_{t-1})^2}\right\|<C, \; 
		E\left\|\dfrac{\bm Y_{t-1} \bm X_{t-1}^\prime}{(\bm \beta_0^\prime \bm X_{t-1})^2}\right\|<C \; \text{and} \; 
		E\left\|\dfrac{\bm X_{t-1} \bm X_{t-1}^\prime}{(\bm \beta_0^\prime \bm X_{t-1})^2}\right\|<C.
		\]
		Thus, $\Sigma$ is finite, and if $E(\eta_t^4)<\infty$ such that $\kappa_1,\kappa_2<\infty$, then $\Omega$ is also finite. 
		
		Let $\bm x = (\bm x^{\prime}_1,\bm x^{\prime}_2)^{\prime}$, where $\bm x_1\in\mathbb{R}^p$ and $\bm x_2\in\mathbb{R}^{2p+1}$ are arbitrary non-zero constant vectors. It follows that
		\begin{align}
		\bm x^{\prime}\Omega \bm x 
		&= E\left\lbrace\dfrac{(\bm x^{\prime}_1\bm Y_{t-1})^2+\kappa_2(\bm x^{\prime}_2\bm X_{t-1})^2+2\kappa_1\bm x^{\prime}_2\bm X_{t-1}\bm Y^{\prime}_{t-1}\bm x_1}{(\bm \beta^{\prime}_{0}\bm X_{t-1})^2}\right\rbrace \notag\\
		&= E\left\lbrace\dfrac{(\bm x^{\prime}_1\bm Y_{t-1}+\kappa_1\bm x^{\prime}_2\bm X_{t-1})^2+(\kappa_2-\kappa_1^2)(\bm x^{\prime}_2\bm X_{t-1})^2}{(\bm \beta^{\prime}_{0}\bm X_{t-1})^2}\right\rbrace. \label{positive Omega}
		\end{align}
		By Cauchy-Schwarz inequality, $\kappa_1^2=\left[\cov(\eta_t,\eta^2_t)\right]^2\leq \var(\eta_t)\var(\eta^2_t)=\kappa_2$, and the equality holds when $P(\eta_t^2-c\eta_t=1)=1$ for any $c\in \mathbb{R}$, which is equivalent to $\det(D)=0$. 
		Since $D$ is positive definite, we have $\kappa_2-\kappa_1^2>0$ and thus $\bm x^{\prime}\Omega \bm x >0$, i.e. $\Omega$ is positive definite. Moreover, by \eqref{constant Vector}, it can be verified that
		\begin{align*}
		\bm x^{\prime}\Sigma \bm x 
		&= E\left\lbrace\dfrac{(\bm x^{\prime}_1\bm Y_{t-1})^2+2(\bm x^{\prime}_2\bm X_{t-1})^2}{(\bm \beta^{\prime}_{0}\bm X_{t-1})^2}\right\rbrace>0.
		\end{align*}
		As a result, $\Sigma$ is positive definite. Hence, (i) holds. 
		
		Note that $\Omega= E\left[\partial \ell_{t}\left(\bm \theta_{0}\right)/\partial \bm \theta \partial \ell_{t}\left(\bm \theta_{0}\right)/\partial \bm \theta^{\prime}\right]$. By the Martingale Central Limit Theorem and the Cram\'{e}r-Wold device, we can show that (ii) holds. 
	\end{proof}
	
	\begin{proof}[Proof of Theorem \ref{thm2QMLE}]
		By Lemma \ref{three o_p}(i) and Lemma \ref{maxmum at true parm}, we have established all the conditions for consistency in Theorem 4.1.1 in \cite{Amemiya1985}, and hence $\widehat{\bm\theta}_n \rightarrow_{p} \bm\theta_{0}$ as $n\to\infty$. 
		
		By Lemma \ref{three o_p}(iii), for any $\bm\theta=\bm\theta_0+o_p(1)$, we have $n^{-1}\sum_{t=p+1}^{n}\partial^2\ell_t(\bm\theta)/\partial\bm\theta\partial\bm\theta^{\prime}=-\Sigma+o_p(1)$.  
		By Taylor's expansion and the consistency of $\widehat{\bm\theta}_n$, then we have
		\begin{equation}\label{QMLE representation}
		\sqrt{n}(\widehat{\bm\theta}_n-\bm\theta_0) =\Sigma^{-1}\dfrac{1}{\sqrt{n}}\sum_{t=p+1}^{n}\dfrac{\partial \ell_t(\bm \theta_0)}{\partial \bm \theta} +o_p(1).
		\end{equation}
		This together with Lemma \ref{MCLT}, we have established all the conditions of Theorem 4.1.3 in \cite{Amemiya1985}, and hence the asymptotic normality follows. 
	\end{proof}
	
	\subsection{Technical details for model selection}
	
	In this section, notations $\Theta^p$, $\bm\theta^{p}$, $\widehat{\bm\theta}_{n}^{p}$ and $\widehat{\Sigma}^{p}$ are employed to emphasize their dependence on the order $p$. 
	We first derive the proposed BICs in Section 3.2, and then establish their consistency in order selection. 
	
	\subsubsection{Derivation of BICs}
	Denote $\bm Y$ as the observed data and $m(\bm Y)$ as its marginal distribution. Let $\pi(p)\ (p\in \{1,2,\ldots,p_{\max}\})$ be a discrete prior over the order set $\{1,2,\ldots,p_{\max}\}$, and $g(\bm\theta^p\mid p)$ be a prior on $\bm\theta^p$ given the order $p$. Moreover, denote $L(\bm \theta^p\mid \bm Y)$ as the likelihood of $\bm Y$ under the model with order $p$, then we have $\ln L(\bm \theta^p\mid \bm Y)=L_n(\bm \theta^p)$. 
	By Bayes' Theorem, the joint posterior of $p$ and $\bm \theta^p$ can be written as
	\[
	P\left(p,\bm \theta^p\mid \bm Y\right)=\dfrac{\pi(p)g(\bm \theta^p\mid p)L(\bm \theta^p\mid \bm Y)}{m(\bm Y)}.
	\]
	Then the posterior probability of $p$ is given by
	\[
	P\left( p\mid \bm Y \right)=m^{-1}(\bm Y)\pi(p)\int L(\bm \theta^p\mid \bm Y)g(\bm \theta^p\mid p) d\bm \theta^p.
	\]
	To maximize $P(p\mid \bm Y)$, it is equivalent to minimize $-2\ln P\left(p\mid\bm Y\right)$ as below
	\[
	-2\ln P(p\mid \bm Y)=2\ln [m(\bm Y)]-2\ln [\pi(p)]-2\ln \left[ \int L\left(\bm \theta^p\mid \bm Y\right)g\left( \bm \theta^p\mid p \right)d\bm\theta^p \right].
	\]
	Consider the noninformative priors for $\bm\theta^p$ and $p$ such that $g\left( \bm \theta^p\mid p \right)=1$ and $\pi(k)=p_{\max}^{-1}$. 
	Since $m(\bm Y)$ is constant with respect to $p$, then we have
	\begin{equation}\label{approx}
	-2\ln P(p\mid \bm Y)\varpropto -2\ln \left[ \int L\left(\bm \theta^p\mid \bm Y\right)d\bm\theta^p \right].
	\end{equation}
	To approximate the above term, we take a second-order Taylor's expansion of the log-likelihood about $\widehat{\bm{\theta}}^p_n$, together with $\widehat{\Sigma}^p=-(n-p)^{-1}\partial^2 \ln L(\widehat{\bm{\theta}}_n^p\mid \bm Y)/(\partial \bm\theta^p \partial \bm {\theta^p}^\prime)$, then it follows that
	\begin{equation*}\label{log-likelihood_approx}
	\begin{aligned}
	\ln L(\bm \theta^p\mid \bm Y)
	\approx &\ln L(\widehat{\bm \theta}^p_n\mid \bm Y)+(\bm \theta^p-\widehat{\bm \theta}_n^p)^\prime \dfrac{\partial \ln L(\widehat{\bm \theta}_n^p\mid \bm Y)}{\partial \bm\theta^p}\\
	&+\dfrac{1}{2}(\bm \theta^p-\widehat{\bm \theta}_n^p)^\prime\left[ \dfrac{\partial^2 \ln L(\widehat{\bm{\theta}}_n^p\mid \bm Y)}{\partial \bm\theta^p \partial \bm {\theta^p}^\prime} \right](\bm \theta^p-\widehat{\bm \theta}_n^p)\\
	=&\ln L(\widehat{\bm \theta}^p_n\mid \bm Y)-\dfrac{1}{2}(\bm \theta^p-\widehat{\bm \theta}_n^p)^\prime\left[(n-p)\widehat{\Sigma}^p\right](\bm \theta^p-\widehat{\bm \theta}_n^p).
	\end{aligned}
	\end{equation*}	
	This implies that
	\[
	L(\bm \theta^p\mid \bm Y)\approx L(\widehat{\bm \theta}^p_n\mid \bm Y)\exp\left\lbrace -\dfrac{1}{2}(\bm \theta^p-\widehat{\bm \theta}_n^p)^\prime\left[(n-p)\widehat{\Sigma}^p\right](\bm \theta^p-\widehat{\bm \theta}_n^p) \right\rbrace.
	\]
	Similar to the Laplace method, we then have the following approximation for the integral
	\begin{align*}
	\int L(\bm\theta^p\mid \bm Y)d\bm \theta^p
	&\approx L(\widehat{\bm \theta}^p_n\mid \bm Y)\int \exp\left\lbrace -\dfrac{1}{2}(\bm \theta^p-\widehat{\bm \theta}^p_n)^\prime\left[(n-p)\widehat{\Sigma}^p\right](\bm \theta^p-\widehat{\bm \theta}^p_n) \right\rbrace d\bm\theta^p\\
	&=L(\widehat{\bm \theta}^p_n\mid \bm Y)(2\pi)^{\frac{3p+1}{2}}\left[\det\left((n-p)\widehat{\Sigma}^p\right)\right]^{-\frac{1}{2}}\\
	&=L(\widehat{\bm \theta}^p_n\mid \bm Y)\left(\dfrac{2\pi}{n-p}\right)^{\frac{3p+1}{2}}\left[\det(\widehat{\Sigma}^p)\right]^{-\frac{1}{2}}.
	\end{align*}
	This together with \eqref{approx} and $\ln L(\widehat{\bm \theta}^p_n\mid \bm Y)=L_n(\widehat{\bm \theta}_n^p)$, implies that
	\begin{align*}
	-2\ln P(p\mid \bm Y)
	&\approx -2L_n(\widehat{\bm \theta}^p_n)+(3p+1)\ln \left(\dfrac{n-p}{2\pi}\right)+\ln\left[ \det\left( \widehat{\Sigma}^p \right)\right]=\text{BIC}_2(p).
	\end{align*}
	As $n\to\infty$, ignoring the $O(1)$ terms in the above approximation, we obtain that
	\[
	-2\ln P(p\mid \bm Y)\approx -2L_n(\widehat{\bm \theta}^p_n)+(3p+1)\ln (n-p)=\text{BIC}_1(p).
	\]	
	Motivated by the above approximations, the BICs are defined as in \eqref{BIC1} and \eqref{BIC2}.
	
	\subsubsection{Proof of Theorem \ref{thm3BIC}}
	\begin{proof}
		Denote the true order of model \eqref{model} as $p_0$. Note that $\text{BIC}_2(p) = \text{BIC}_1(p)-(3p+1)\ln(2\pi)+\ln(\det(\widehat{\Sigma}^{p}))$, and $-(3p+1)\ln(2\pi)+\ln(\det(\widehat{\Sigma}^{p}))$ is bounded as $n\to \infty$.  
		It suffices to show that, for any $p\neq p_0$, 
		\begin{equation}\label{goal of BIC}
		\lim_{n\rightarrow \infty} P(\text{BIC}_1(p)-\text{BIC}_1(p_0)>0)=1.
		\end{equation}
		
		We first consider the case with $p > p_0$, i.e. the model is overfitted. 
		Note that the model with order $p$ corresponds to a bigger model, and it holds that 
		\[\ell_t(\bm\theta_{0}^{p}) =\ell_t(\bm\theta_{0}^{p_0}) \quad\text{and}\quad L_n(\bm\theta_{0}^{p}) =L_n(\bm\theta_{0}^{p_0}).\]
		Denote $\bm\nu=\sqrt{n}(\widehat{\bm\theta}^{p_0}_{n}-\bm\theta_0^{p_0})$. Similar to the proof of Theorem \ref{thm3BIC}, we can show that $\bm\nu=O_p(1)$. By Taylor's expansion and Slutsky's theorem, it can be shown that
		\begin{align}\label{Ln distance}
		L_n(\widehat{\bm\theta}_{n}^{p_0})-L_n(\bm\theta_{0}^{p_0}) 
		&= \bm\nu^{\prime}\dfrac{1}{\sqrt{n}}\dfrac{\partial L_n(\bm\theta_{0}^{p_0})}{\partial \bm \theta}-\bm\nu^{\prime}\Sigma^{p_0}\bm\nu+o_p(1) = O_p(1).
		\end{align}
		Similarly, it can be verified that $L_n(\widehat{\bm\theta}_{n}^{p})-L_n(\bm\theta_{0}^{p})=O_p(1)$. As a result,  
		\[L_n(\widehat{\bm\theta}_{n}^{p})-L_n(\widehat{\bm\theta}_{n}^{p_0}) = [L_n(\widehat{\bm\theta}_{n}^{p})-L_n(\bm\theta_{0}^{p})]-[L_n(\widehat{\bm\theta}_{n}^{p_0})-L_n(\bm\theta_{0}^{p_0})] + [L_n(\bm\theta_{0}^{p})-L_n(\bm\theta_{0}^{p_0})]= O_p(1).\]
		Hence, we have
		\begin{align*}
		&\text{BIC}_1(p)-\text{BIC}_1(p_0) \\
		=& -2[L_n(\widehat{\bm \theta}_{n}^{p})-L_n(\widehat{\bm \theta}_{n}^{p_0})] +\left[(3p+1)\ln (n-p)-(3p_0+1)\ln (n-p_0)\right]\\
		< & O_p(1) + 3(p-p_0)\ln (n-p_0) \to \infty
		\end{align*}
		as $n\to \infty$. Therefore, \eqref{goal of BIC} holds for $p > p_0$. 
		
		We next consider the case with $p < p_0$, i.e. the model is underfitted. 
		Let $\bm\theta_0^p = \arg\max_{\bm\theta \in \Theta^p} E[\ell_t(\bm\theta^p)]$. 
		Similar to the proof of Theorem \ref{thm3BIC} and \eqref{Ln distance}, we can verify that  $\sqrt{n}(\widehat{\bm\theta}_{n}^{p}-\bm\theta_{0}^{p})=O_p(1)$ and  
		\[L_n(\widehat{\bm\theta}_{n}^{p})-L_n(\bm\theta_{0}^{p})=O_p(1).\]
		Since the model with order $p$ corresponds to a smaller model, we have $E\ell_t(\bm\theta_0^{p_0})\geq E\ell_{t}(\bm \theta_0^{p})+\epsilon$ for some positive constant $\epsilon$. 
		By ergodic theorem, we have $n^{-1}L_n(\bm\theta_{0}^p)= E\ell_{t}(\bm \theta_0^p)+o_p(1)$. Thus it holds that
		\begin{equation*}
		L_n(\bm\theta_{0}^{p})-L_n(\bm\theta_{0}^{p_0})=-n\{E[\ell_t(\bm\theta_0^{p_0})]- E[\ell_{t}(\bm\theta_0^{p})]\}+ o_p(n)=-n\epsilon + o_p(n).
		\end{equation*}
		Therefore, we have
		\begin{align*}
		L_n(\widehat{\bm\theta}_{n}^{p})-L_n(\widehat{\bm\theta}_{n}^{p_0}) &= [L_n(\widehat{\bm\theta}_{n}^{p})-L_n(\bm\theta_{0}^{p})]-[L_n(\widehat{\bm\theta}_{n}^{p_0})-L_n(\bm\theta_{0}^{p_0})] + [L_n(\bm\theta_{0}^{p})-L_n(\bm\theta_{0}^{p_0})] \\
		&=O_p(1)-n\epsilon + o_p(n).
		\end{align*}
		This together with $(3p+1)\ln (n-p)-(3p_0+1)\ln (n-p_0)=O(\ln n)$, implies that
		\begin{align*}
		&\text{BIC}_1(p)-\text{BIC}_1(p_0) \\
		=& -2[L_n(\widehat{\bm \theta}_{n}^{p})-L_n(\widehat{\bm \theta}_{n}^{p_0})] +\left[(3p+1)\ln (n-p)-(3p_0+1)\ln (n-p_0)\right]\\
		= & 2n\epsilon + o_p(n) + O_p(1) + O(\ln n) \to \infty
		\end{align*}
		as $n\to \infty$. Hence, \eqref{goal of BIC} holds for $p < p_0$. 
		The proof is accomplished. 
	\end{proof}
	
	\subsection{Proof of Theorem \ref{thm test asymmetry}}
	
	\begin{proof}
		The hypotheses of the asymmetry test are $H_0: R\bm\theta_{0} =\bm 0_{p}\; \text{versus}\; H_1: R\bm\theta_{0} \neq\bm 0_{p}$, where $R = (0_{p\times(p+1)}, I_{p}, -I_{p})$ is the $p\times(3p+1)$ matrix, $\bm\theta_{0}$ is the true parameter vector and $\bm 0_{p}$ is a $p$-dimensional zero vector. 
		By Lemma \ref{three o_p}, under $H_0$, we have
		\begin{equation}\label{matrix consistency}
		\widehat{\Sigma}=\Sigma+o_p(1),\; \widehat{\Xi}=\Xi+o_p(1),\; \widetilde{\Sigma}=\Sigma+o_p(1)\; \text{and} \; \widetilde{\Xi}=\Xi+o_p(1).
		\end{equation}
		We show the null distributions of the Wald, LM and QLR test statistics, respectively. \\	
		\textbf{(i) Wald test} 
		
		By \eqref{QMLE representation}, we have
		$$\sqrt{n} R(\widehat{\bm \theta}_n-\bm \theta_0) =  R\Sigma^{-1}\Omega^{\frac{1}{2}}\left[ \Omega^{-\frac{1}{2}}\dfrac{1}{\sqrt{n}}\sum_{t=p+1}^{n}\dfrac{\partial \ell_t(\bm \theta_0)}{\partial \bm \theta}\right] +o_p(1).$$
		Then under $H_0$, by Lemma \ref{MCLT}, it follows that 
		$$(R{\Xi}R^{\prime})^{-1 / 2} \sqrt{n}R\widehat{\bm \theta}_n \rightarrow_\mathcal{L} N\left(\bm 0, I_{p}\right),$$
		where $\Xi = \Sigma^{-1}\Omega\Sigma^{-1}$. Then by Slutsky's theorem, under $H_0$, it holds that
		$$W_n = n(R\widehat{\bm \theta}_n)^{\prime}(R\widehat{\Xi}R^{\prime})^{-1}R\widehat{\bm \theta}_n=(\sqrt{n}R\widehat{\bm \theta}_n)^{\prime}(R\Xi R^{\prime})^{-1}(\sqrt{n}R\widehat{\bm\theta}_n)+o_p(1)\rightarrow_\mathcal{L} \chi^2_p.$$
		This completes the first part.\\
		\textbf{(ii) Lagrange Multiplier (LM) test} 
		
		Let $\bm k\in R^p$ be a  vector. Under $H_0$, the lagrangian can be formulated as
		\[
		\mathbb{L}(\bm\theta,\bm k) = L_n(\bm \theta)-\bm k^{\prime}R\bm \theta.
		\]
		Denote $(\widetilde{\bm \theta}_n,\widetilde{\bm k}) = \arg\inf_{\bm\theta,\bm k}\mathbb{L}(\bm\theta,\bm k)$, where $\widetilde{\bm \theta}_n$ is the restricted QMLE under $H_0$, and $\widetilde{\bm k}$ is the lagrangian multiplier. Taking the first derivatives of $\mathbb{L}(\bm\theta,\bm k)$ with respect to $\bm\theta$ and $\bm k$ at $(\widetilde{\bm \theta}_n,\widetilde{\bm k})$ respectively, we have
		\begin{equation}\label{notation trans}
		\dfrac{\partial L_n(\widetilde{\bm\theta}_n)}{\partial \bm \theta}=R^{\prime}\widetilde{\bm k}\quad \text{and} \quad R\widetilde{\bm \theta}_{n} = \bm 0_p.
		\end{equation}
		Under $H_0$, by Taylor's expansion and \eqref{matrix consistency}, it holds that
		\begin{equation*}
		\sqrt{n}(\widetilde{\bm\theta}_n-\bm\theta_0) =\Sigma^{-1}\left[\dfrac{1}{\sqrt{n}}\sum_{t=p+1}^{n}\dfrac{\partial \ell_t(\bm \theta_0)}{\partial \bm \theta}-\dfrac{1}{\sqrt{n}}\sum_{t=p+1}^{n}\dfrac{\partial \ell_t(\widetilde{\bm\theta}_n)}{\partial \bm \theta}\right] +o_p(1).
		\end{equation*}
		This together with \eqref{QMLE representation} and \eqref{notation trans}, implies that
		\begin{equation}\label{LMthm1}
		\sqrt{n}(\widehat{\bm \theta}_n-\widetilde{\bm \theta}_{n}) = \Sigma^{-1}\dfrac{1}{\sqrt{n}}\frac{\partial L_n(\widetilde{\bm\theta}_n)}{\partial \bm \theta}+o_p(1)=\Sigma^{-1}\dfrac{1}{\sqrt{n}}R^{\prime}\widetilde{\bm k}+o_p(1).
		\end{equation}
		Under $H_0$, note that $R\widetilde{\bm \theta}_{n} = \bm 0_p$ by \eqref{notation trans} and $R\bm\theta_0 = \bm 0_p$, then by Theorem \ref{thm2QMLE}, we have 
		\[\sqrt{n}R(\widehat{\bm \theta}_n-\widetilde{\bm \theta}_{n})=\sqrt{n}R\widehat{\bm \theta}_n=\sqrt{n}R(\widehat{\bm \theta}_n-{\bm \theta}_0)\rightarrow_{\mathcal{L}}N(\bm 0,R \Xi R^{\prime}) \; \text{as} \; n\to \infty.\] 
		Recall that $\Delta=R\Sigma^{-1}R^\prime$. Multiplying both sides of \eqref{LMthm1} by $R$, we have 
		\begin{equation}\label{ktilde rep}
		\dfrac{1}{\sqrt{n}}\widetilde{\bm k} = \Delta^{-1}\sqrt{n}R(\widehat{\bm \theta}_n-\bm\theta_0)+o_p(1)\rightarrow_{\mathcal{L}}N(\bm 0,\Delta^{-1}R \Xi R^{\prime}\Delta^{-1}) \; \text{as} \; n\to \infty.
		\end{equation}
		This together with \eqref{LMthm1} and Slutsky's theorem, implies that, under $H_0$,
		\begin{align*}
		L_n&=\dfrac{1}{n}\dfrac{\partial L_n(\widetilde{\bm \theta}_{n})}{\partial \bm \theta^{\prime}}\widetilde{\Sigma}^{-1}R^\prime(R\widetilde{\Xi}R^{\prime})^{-1}R\widetilde{\Sigma}^{-1}\dfrac{\partial L_n(\widetilde{\bm \theta}_{n})}{\partial \bm \theta}\\
		&=\dfrac{1}{\sqrt{n}}\widetilde{\bm k}^\prime (R{\Sigma}^{-1}R^\prime)(R\Xi R^{\prime})^{-1}(R{\Sigma}^{-1}R^\prime) \dfrac{1}{\sqrt{n}}\widetilde{\bm k}+o_p(1)\rightarrow_{\mathcal{L}} \chi^2_p \; \text{as} \; n\to \infty.
		\end{align*}
		This completes the second part.\\
		\textbf{(iii) Quasi-likelihood ratio (QLR) test} 
		
		By Taylor's expansion of $L_{n}(\widetilde{\bm{\theta}}_{n})$ about $\widehat{\bm{\theta}}_{n}$, together with $\partial L_n(\widehat{\bm \theta}_n)/\partial \bm \theta=\bm 0$ and \eqref{matrix consistency}-\eqref{ktilde rep}, under $H_0$, we have 
		\begin{align*}
		Q_n&=-2[L_{n}(\widetilde{\bm \theta}_{n})-L_{n}(\widehat{\bm \theta}_{n})] \\ 
		&=-2\dfrac{\partial L_{n}(\widehat{\bm \theta}_{n})}{\partial \bm \theta }(\widetilde{\bm \theta}_{n}-\widehat{\bm \theta}_{n})+n(\widetilde{\bm \theta}_{n}-\widehat{\bm \theta}_{n})^{\prime} \Sigma(\widetilde{\bm \theta}_{n}-\widehat{\bm \theta}_{n})+o_{p}(1)\\
		&=\dfrac{1}{\sqrt{n}}\widetilde{\bm k}^\prime R \Sigma^{-1} R^\prime \dfrac{1}{\sqrt{n}}\widetilde{\bm k}+o_p(1)\\  
		&= \left[(\Delta^{-1}R \Xi R^{\prime}\Delta^{-1})^{\frac{1}{2}}\dfrac{1}{\sqrt{n}}\widetilde{\bm k}\right]^{\prime}\Psi\left[(\Delta^{-1}R \Xi R^{\prime}\Delta^{-1})^{\frac{1}{2}}\dfrac{1}{\sqrt{n}}\widetilde{\bm k}\right]+o_p(1)\rightarrow_\mathcal{L} \sum_{j=1}^{p}e_jx_j
		\end{align*}
		as $n\to\infty$, where $\Psi=\Delta^{-1/2}R\Xi R^\prime\Delta^{-1/2}$, $e_j$ is the $j$-th eigenvalue of $\Psi$, and $x_j$'s are the $i.i.d.$ random variables following the $\chi^2_1$ distribution. This completes the third part. 
	\end{proof}
	
	\subsection{Proof of Theorem \ref{LAN}}
	
	Recall that the time series $\{y_{p+1,n},\ldots,y_{n,n}\}$ are generated by
	\begin{equation*}
	H_{1n}:\quad y_{t,n} = \left(\bm\alpha_0+\dfrac{\bm h_{\alpha}}{\sqrt{n}}\right)^\prime \bm Y_{t-1,n} + \eta_t\left(\bm\beta_0+\dfrac{\bm h_{\beta}}{\sqrt{n}}\right)^\prime \bm X_{t-1,n}, 
	\end{equation*}
	where $\bm Y_{t,n} = (y_{t,n},\ldots,y_{t-p+1,n})^{\prime}$ and $\bm X_{t,n} =(1,\bm Y^{\prime}_{t+,n},-\bm Y^{\prime}_{t-,n})^{\prime}$ with $\bm Y_{t+,n} = (y_{t,n}^{+},\ldots,\\ y_{t-p+1,n}^{+})^{\prime}$ and $\bm Y_{t-,n} = (y_{t,n}^{-},\ldots,y_{t-p+1,n}^{-})^{\prime}$. Moreover, 
	\[\ell_{t,n}(\bm \theta)=-\ln \left(\bm \beta^{\prime} \bm X_{t-1,n}\right)-\dfrac{\left(y_{t,n}-\bm \alpha^{\prime} \bm Y_{t-1,n}\right)^{2}}{2\left(\bm{\beta}^{\prime} \bm{X}_{t-1,n}\right)^{2}} \quad \text{and}\quad L_{n,h}(\bm \theta)=\sum_{i=p+1}^{n}\ell_{t,n}(\bm \theta).\]
	To prove Theorem \ref{LAN}, the following two lemmas are required. 
	
	\begin{lemma}\label{triangular array lemma}
		Let $\{X_{ni}:1\leq i\leq k_n,n=1,2,\ldots\}$ be a mean-zero triangular array of $\beta$-mixing sequences that are $L^s$-bounded for some $s>1$, and $\mathcal{F}_{ni}=\sigma(X_{n1},\ldots,X_{ni})$ for $1\leq i \leq k_n$. Then $\{X_{ni},\mathcal{F}_{ni}\}$ is a uniformly integrable $L^1$-mixing. Furthermore, $E|k_n^{-1}\sum_{i=1}^{k_n}X_{ni}|\rightarrow 0$ as $n\rightarrow\infty$, and hence $k_n^{-1}\sum_{i=1}^{k_n}X_{ni}\rightarrow 0$ in probability as $n\rightarrow \infty$.
	\end{lemma}
	\begin{proof}
		See Example 4 in Section 3 of \cite{Andrews1988}.
	\end{proof}
	
	\begin{lemma}\label{LAN-lemma2}
		Suppose Assumptions \ref{assumption_compact}-\ref{assum_y_tH1} and $E(\eta_t^4)<\infty$ hold. Then under $\mathbb{P}_{n,h}$, for any sequence such that $\bm\theta^\ast \to_p \bm\theta_{0}$, it holds that
		\[(i)\ \left\Arrowvert -\dfrac{1}{n}\sum_{t=p+1}^n\dfrac{\partial^2\ell_{t,n}(\bm \theta^\ast)}{\partial \bm\theta \partial \bm\theta^\prime}-\Sigma\right\Arrowvert=o_p(1); \; (ii)\ \left\Arrowvert \dfrac{1}{n} \sum_{t=p+1}^n\dfrac{\partial\ell_{t,n}(\bm \theta^\ast)}{\partial \bm\theta}\dfrac{\partial\ell_{t,n}(\bm\theta^\ast)}{\partial \bm \theta^\prime}-\Omega \right\Arrowvert=o_p(1);\]
		where $\Omega$ and $\Sigma$ are defined as in Theorem \ref{thm2QMLE}.
	\end{lemma}
	\begin{proof}
		Under $\mathbb{P}_{n,h}$, we have
		\begin{align*}
		&\dfrac{\partial \ell_{t,n}(\bm\theta)}{\partial \bm\alpha}=\dfrac{\bm Y_{t-1,n}(y_{t,n}-\bm \alpha^\prime \bm Y_{t-1,n})}{(\bm \beta^\prime \bm X_{t-1,n})^2}, \; 
		\dfrac{\partial \ell_{t,n}(\bm \theta)}{\bm \beta}=-\dfrac{\bm X_{t-1,n}}{\bm \beta^\prime \bm X_{t-1,n}}\left[ 1-\dfrac{(y_{t,n}-\bm \alpha^\prime\bm X_{t-1,n})^2}{(\bm \beta^\prime \bm X_{t-1,n})^2} \right],\\
		&\dfrac{\partial^2 \ell_{t,n}(\bm\theta)}{\partial \bm\alpha \partial \bm\alpha^\prime}=-\dfrac{\bm Y_{t-1,n}\bm Y_{t-1,n}^\prime}{(\bm\beta^\prime \bm X_{t-1,n})^2}, \; 
		\dfrac{\partial^2 \ell_{t,n}(\bm\theta)}{\partial \bm\beta \partial \bm\beta^\prime}=\dfrac{\bm X_{t-1,n}\bm X_{t-1,n}^\prime}{(\bm\beta^\prime \bm X_{t-1,n})^2}\left[ 1-\dfrac{3(y_{t,n}-\bm\alpha^\prime \bm Y_{t-1,n})^2}{(\bm\beta^\prime \bm X_{t-1,n})^2} \right], \\
		&\dfrac{\partial^2 \ell_{t,n}(\bm\theta)}{\partial \bm\alpha \partial \bm\beta^\prime}=-\dfrac{2\bm Y_{t-1,n}\bm X_{t-1,n}^\prime(y_{t,n}-\bm\alpha^\prime \bm Y_{t-1,n})}{(\bm\beta^\prime \bm X_{t-1,n})^3}.
		\end{align*}
		By Assumption \ref{assum_y_tH1}, using similar arguments as for Lemma \ref{boundedness}, we can show that
		\[
		E\sup_{\bm \theta \in \Theta}\left|\ell_{t,n}(\bm \theta)\right|<\infty; \quad  E\sup_{\bm \theta \in \Theta} \left\| \dfrac{\partial \ell_{t,n}(\bm \theta)}{\partial \bm \theta } \right\|<\infty; \quad  E\sup _{\bm \theta \in \Theta}\left\|\dfrac{\partial^{2} \ell_{t,n}(\bm \theta)}{\partial \bm \theta \partial \bm \theta^{\prime}}\right\|<\infty.
		\]
		Hence, by the similar arguments as for \eqref{uniformal converge}, we can show that 
		\[\sup_{\bm \theta \in \Theta}\left\Arrowvert -\dfrac{1}{n} \sum_{t=p+1}^n\dfrac{\partial^2\ell_{t,n}(\bm \theta)}{\partial \bm\theta \partial \bm\theta^\prime}-\Sigma\right\Arrowvert=o_p(1), \; \sup_{\bm \theta \in \Theta}\left\Arrowvert \dfrac{1}{n} \sum_{t=p+1}^n\dfrac{\partial\ell_{t,n}(\bm \theta)}{\partial \bm\theta}\dfrac{\partial\ell_{t,n}(\bm\theta)}{\partial \bm \theta^\prime}-\Omega \right\Arrowvert=o_p(1),\]
		which implies that (i) and (ii) hold. This completes the proof of this lemma. 
	\end{proof}
	
	\begin{proof}[Proof of Theorem \ref{LAN}]
		Below we will show that, as $n\to \infty$, (i) $\widehat{\bm\theta}_{n,h}\to_p \bm\theta_0$ and (ii) $\sqrt{n}(\widehat{\bm\theta}_{n,h}-\bm\theta_0)\rightarrow_{\mathcal{L}}N(\bm h,\Xi)$, where $\bm h=(\bm h_{\alpha}^{\prime},\bm h_{\beta}^{\prime})^\prime$ and $\Xi=\Sigma^{-1}\Omega\Sigma^{-1}$.\\
		(i) Similar to the proof of Lemma \ref{boundedness}, by Assumptions \ref{assumption_compact} and \ref{assum_y_tH1}, we can prove that
		\begin{equation}\label{finite_ell^2_H1}
		E\sup_{\bm \theta \in \Theta}\ell^2_{t,n}(\bm\theta)<\infty.
		\end{equation}
		By Assumption \ref{assum_y_tH1}, $\{y_{t,n}\}$ is geometrically ergodic, and hence $\beta$-mixing. By Theorem 3.49 of \cite{White2001}, it follows that $\ell_{t,n}(\bm\theta)$ is $\beta$-mixing. Then, by Lemma \ref{triangular array lemma} with $s=2$, we have
		\begin{equation}\label{point converge to E}
		\lim _{n \rightarrow \infty} \dfrac{1}{n-p} \sum_{t=p+1}^{n}\left[ \ell_{t, n}(\bm \theta)-E\ell_{t, n}(\bm \theta)\right]=o_p(1).
		\end{equation}
		Furthermore, the stationarity of $\ell_{t,n}(\bm \theta)$ ensures that
		\[
		\dfrac{1}{n-p} \sum_{t=p+1}^{n} E\ell_{t, n}(\bm \theta)=E\ell_{t, n}(\bm\theta),
		\]
		and the dominated convergence theorem entails that
		\[
		\lim _{n \rightarrow \infty} E\ell_{t, n}(\bm\theta)=E \lim _{n \rightarrow \infty} \ell_{t, n}(\bm\theta)=E\ell_{t}(\bm\theta).
		\]
		These together with \eqref{point converge to E} imply that, for any $\bm \theta \in \Theta$,
		\begin{equation}\label{point converge}
		\dfrac{1}{n-p}\sum_{t=p+1}^{n}\ell_{t,n}(\bm\theta)-E\ell_{t}(\bm\theta)=o_p(1).
		\end{equation}
		Moreover, by similar arguments as for Lemma \ref{boundedness}(ii), we can prove that 
		$$\sup_{\bm \theta\in \Theta}\left\| \dfrac{1}{n-p}\sum_{t=p+1}^{n}\dfrac{\partial \ell_{t, n}(\bm \theta)}{\partial \bm\theta} \right\|=O_p(1).$$ 
		As a result, Assumption 3A in \cite{nelson1991conditional} holds. Then by \eqref{point converge} and Theorem 2.1 of \cite{nelson1991conditional}, we have
		\begin{equation}\label{uniformal converge}
		\sup_{\bm \theta \in \Theta}\left| \dfrac{1}{n-p}\sum_{t=p+1}^{n}\ell_{t, n}(\bm \theta)-E\ell_{t}(\bm\theta) \right|=o_p(1).
		\end{equation}
		Finally, since $E\ell_{t}(\bm\theta)$ attains its global maximum at $\bm\theta_0$ by Lemma \ref{maxmum at true parm}, (i) holds by \eqref{uniformal converge} and Theorem 4.1.1 in \cite{Amemiya1985}. \\
		(ii) By Taylor's expansion, we have
		\begin{equation*}
		\sum_{t=p+1}^{n}\dfrac{\partial \ell_{t,n}(\widehat{\bm\theta}_{n,h})}{\partial \bm\theta}
		=\sum_{t=p+1}^{n}\dfrac{\partial \ell_{t,n}({\bm\theta}_{n})}{\partial \bm\theta}
		+\sum_{t=p+1}^{n}\dfrac{\partial^2 \ell_{t,n}(\bm\theta^\ast)}{\partial \bm\theta \partial \bm\theta^\prime}(\widehat{\bm{\theta}}_{n,h}-\bm\theta_n),
		\end{equation*}
		where $\bm\theta^\ast$ is between $\widehat{\bm\theta}_{n,h}$ and $\bm\theta_n$. 
		Moreover, by the Martingale Central Limit Theorem and the Cram\'{e}r-Wold device, together with Lemma \ref{LAN-lemma2}(ii), we can show that, as $n\to \infty$, 
		\[\dfrac{1}{\sqrt{n}}\sum_{t=p+1}^{n}\dfrac{\partial \ell_{t,n}({\bm\theta}_{n})}{\partial \bm\theta} \to_\mathcal{L}N(0,\Xi).\] 
		Recall that $\bm\theta_{n}=\bm\theta_{0}+\bm h/\sqrt{n}$. Then by (i) and Lemma \ref{LAN-lemma2}(i), we have
		\begin{align*}
		\sqrt{n}(\widehat{\bm \theta}_{n,h}-\bm\theta_0)&=\bm h+\left[ -\dfrac{1}{n}\sum_{t=p+1}^{n}\dfrac{\partial^2\ell_{t,n}(\bm\theta^\ast)}{\partial\bm\theta \partial\bm\theta^\prime} \right]^{-1}\dfrac{1}{\sqrt{n}}\sum_{t=p+1}^{n}\dfrac{\partial \ell_{t, n}(\bm\theta_{n})}{\partial\bm\theta} \\
		& \to_\mathcal{L}N(\bm h,\Xi) \quad \text{as}\; n\to\infty.
		\end{align*}
		The proof of this theorem is accomplished. 
	\end{proof}
	
	\subsection{Proof of Theorem \ref{thm test asymmetryH1}}
	
	\begin{proof}
		Denote $\Xi(\bm\theta)=[\Sigma(\bm\theta)]^{-1}\Omega(\bm\theta)[\Sigma(\bm\theta)]^{-1}$, where
		\[\Sigma(\bm \theta)=-E\left[\dfrac{\partial^{2} \ell_{t}\left(\boldsymbol{\theta}\right)}{\partial \boldsymbol{\theta} \partial \boldsymbol{\theta}^{\prime}}\right], \; 
		\Omega(\bm \theta)=E\left[\frac{\partial \ell_{t}\left(\boldsymbol{\theta}\right)}{\partial \boldsymbol{\theta}} \frac{\partial \ell_{t}\left(\boldsymbol{\theta}\right)}{\partial \boldsymbol{\theta}^{\prime}}\right].\]
		Clearly, $\Sigma(\bm\theta_0)=\Sigma$, $\Sigma(\widehat{\bm \theta}_n)=\widehat{\Sigma}$, $\Omega(\bm\theta_0)=\Omega$, $\Omega(\widehat{\bm \theta}_n)=\widehat{\Omega}$, $\Xi(\bm\theta_0)=\Xi$, $\Xi(\widetilde{\bm \theta}_n)=\widetilde{\Xi}$ and $\Xi(\widehat{\bm \theta}_n)=\widehat{\Xi}$. 
		Based on Theorem \ref{LAN}, we show the limit distributions of the Wald, LM and QLR test statistics under $\mathbb{P}_{n,h}$, respectively.\\
		\textbf{(i) Wald test}
		
		Since $R\bm\theta_{0} =\bm 0_{p}$, under $\mathcal{P}_{n,h}$, by Theorem \ref{LAN}, we have
		\begin{equation}\label{Centerned QMLE H1}
		\left(R \Xi R^\prime\right)^{-1/2}\sqrt{n}R\widehat{\bm \theta}_{n,h}\rightarrow_\mathcal{L} N(\left(R \Xi R^\prime\right)^{-1/2}R\bm h,I_p) \quad \text{as}\; n\to\infty.
		\end{equation}
		Then by Slutsky's theorem and $\widehat{\bm\theta}_{n,h}\to_p \bm\theta_0$, it holds that
		\[
		W_n = n(R\widehat{\bm \theta}_{n,h})^{\prime}[R{\Xi(\widehat{\bm \theta}_{n,h})}R^{\prime}]^{-1}R\widehat{\bm \theta}_{n,h}=(\sqrt{n}R\widehat{\bm \theta}_{n,h})^{\prime}\left(R \Xi R^\prime\right)^{-1}(\sqrt{n}R\widehat{\bm \theta}_{n,h})+o_p(1)\rightarrow_\mathcal{L} \chi^2_p(\delta).
		\]
		where $\delta=\bm h^{\prime} R^{\prime}\left(R \Xi R^\prime\right)^{-1} R \bm h$.\\
		\textbf{(ii) Lagrange Multiplier test}
		
		Similar to the proof of Theorem \ref{thm test asymmetry}(ii) and Theorem \ref{LAN}, we can show that 
		\begin{equation}\label{Restricted QMLE H1}
		\dfrac{1}{\sqrt{n}}\dfrac{\partial L_{n,h}(\widetilde{\bm \theta}_{n,h})}{\partial \bm \theta}=R^{\prime}\Delta^{-1}\sqrt{n}R\widehat{\bm \theta}_{n,h}+o_p(1)
		\end{equation}
		as $n\to\infty$, where $\widetilde{\bm \theta}_{n,h}$ is the restricted QMLE under $\mathcal{P}_{n,h}$ such that $R\bm\theta_{0} =\bm 0_{p}$. 
		By Theorem \ref{LAN}, if $R\bm\theta_{0} =\bm 0_{p}$ holds, then we have $\widetilde{\bm \theta}_{n,h}\to_p \bm\theta_0$ as $n\to\infty$.
		Recall that $\Delta=R\Sigma^{-1}R^\prime$. Therefore, by Slutsky's theorem and \eqref{Centerned QMLE H1}, we have
		\begin{align*}
		L_n&=(\sqrt{n}R\widehat{\bm \theta}_{n,h})^\prime\Delta^{-1}R[\Sigma(\widetilde{\bm \theta}_{n,h})]^{-1}R^\prime[R\Xi(\widetilde{\bm \theta}_{n,h})R^{\prime}]^{-1}R[\Sigma(\widetilde{\bm \theta}_{n,h})]^{-1}R^{\prime}\Delta^{-1}(\sqrt{n}R\widehat{\bm \theta}_{n,h}) \\
		&=(\sqrt{n}R\widehat{\bm \theta}_{n,h})^\prime \left(R \Xi R^\prime\right)^{-1}(\sqrt{n}R\widehat{\bm \theta}_{n,h})+o_p(1)\rightarrow_\mathcal{L} \chi^2_p(\delta).
		\end{align*}
		\textbf{(iii) Quasi-likelihood ratio test}
		
		Similar to the proof of Theorem \ref{thm test asymmetry}(iii) and Theorem \ref{LAN}, together with the facts that $\partial L_{n,h}(\widehat{\bm \theta}_{n,h})/\partial \bm \theta=\bm 0$, $\sqrt{n}(\widetilde{\bm \theta}_{n,h}-\widehat{\bm \theta}_{n,h})=-\Sigma^{-1}n^{-1/2}\partial L_{n,h}(\widetilde{\bm\theta}_{n,h})/\partial \bm \theta+o_p(1)$, \eqref{Centerned QMLE H1}  and \eqref{Restricted QMLE H1}, it can be verified that
		\begin{align*}
		Q_n&=-2[L_{n,h}(\widetilde{\bm \theta}_{n,h})-L_{n,h}(\widehat{\bm \theta}_{n,h})] \\ 
		&=-2\dfrac{\partial L_{n,h}(\widehat{\bm \theta}_{n,h})}{\partial \bm \theta}(\widetilde{\bm \theta}_{n,h}-\widehat{\bm \theta}_{n,h})+n(\widetilde{\bm \theta}_{n,h}-\widehat{\bm \theta}_{n,h})^{\prime} \Sigma(\widetilde{\bm \theta}_{n,h}-\widehat{\bm \theta}_{n,h})+o_{p}(1)\\
		&=(\sqrt{n}R\widehat{\bm \theta}_{n,h})^\prime\Delta^{-1}(\sqrt{n}R\widehat{\bm \theta}_{n,h})+o_p(1) \\
		&=[\Gamma\left(R \Xi R^\prime\right)^{-1/2}\sqrt{n}R\widehat{\bm \theta}_{n,h}]^\prime \Gamma D \Gamma^{\prime} [\Gamma\left(R \Xi R^\prime\right)^{-1/2}\sqrt{n}R\widehat{\bm \theta}_{n,h}] \rightarrow_\mathcal{L} \sum_{j=1}^{p}e_jx_{j,v_j^2}
		\end{align*}
		as $n\to\infty$, where $D=\left(R \Xi R^\prime\right)^{1/2}\Delta^{-1}\left(R \Xi R^\prime\right)^{1/2}$, 
		$\Gamma$ is an orthogonal matrix such that $\Gamma D \Gamma^\prime=\diag\{e_1,\ldots,e_p\}$ with $e_j$ being the $j$-th eigenvalue of $D$, $v_j$ is the $j$-th component of $\bm v=\Gamma\left(R \Xi R^\prime\right)^{-1/2}R\bm h$, and $x_{j,v^2_j}$'s are independent random variables following the $\chi^2_{1}(v^2_j)$ distribution for $j=1, \ldots, p$.  
		The proof is completed.
	\end{proof}
	
	\subsection{Proof of Theorem \ref{thmACF}}
	
	\begin{proof}
		Denote $\eta_{t}(\bm\theta)=(y_t-\bm \alpha^\prime \bm Y_{t-1})/(\bm\beta^\prime \bm X_{t-1})$. 
		Recall that $\eta_{t}=\eta_{t}(\bm \theta_{0})$, $\widehat{\eta}_{t}=\eta_{t}(\widehat{\bm \theta}_n)$, $E(\eta_t)=0$ and $\var(\eta_t)=1$. 
		When model \eqref{model} is correctly specified, by the ergodic theorem and the dominated convergence theorem, it can be shown that, 
		as $n\to\infty$, 
		\[\bar{\eta}_1=\dfrac{1}{n-p}\sum_{t=p+1}^{n}\widehat{\eta}_t \to_p 0 \quad \text {and} \quad 
		\dfrac{1}{n}\sum_{t=p+1}^{n}\left(\widehat{\eta}_{t}-\bar{\eta}_1\right)^{2} \to_p 1,\]
		\[
		\bar{\eta}_2=\dfrac{1}{n-p}\sum_{t=p+1}^{n}|\widehat{\eta}_t| \to_p \tau_2 \quad \text {and} \quad 
		\dfrac{1}{n}\sum_{t=p+1}^{n}(|\widehat{\eta}_{t}|-\bar{\eta}_2)^2 \to_p {\sigma}^{2}_\xi.
		\]
		Hence, it follows that
		\begin{equation}\label{relationship between hat_acf and tilde_acf}
		\sqrt{n}\left(\widehat{\bm \rho}^{\prime}, \widehat{\bm \gamma}^{\prime}\right)^{\prime}=\sqrt{n}\left(\widetilde{\bm \rho}^{\prime}, \widetilde{\bm\gamma}^{\prime}\right)^{\prime}+o_{p}(1),
		\end{equation}
		where $\widetilde{\bm\rho}=(\widetilde{\rho}_{1},\ldots,\widetilde{\rho}_{M})^\prime$ and $\widetilde{\bm \gamma}=(\widetilde{\gamma}_{1},\ldots,\widetilde{\gamma}_{M})^\prime$ with $\widetilde{\rho}_{k}=n^{-1}\sum_{t=p+k+1}^{n}\widehat{\eta}_{t}\widehat{\eta}_{t-k}$ and $\widetilde{\gamma}_{k}=(n\sigma^2_\xi)^{-1}\sum_{t=p+k+1}^{n}(|\widehat{\eta}_{t}|-\tau_2)(|\widehat{\eta}_{t-k}|-\tau_2)$. 
		It can be verified that
		\begin{align*}
		\sqrt{n}\widetilde{\rho}_{k} &= \dfrac{1}{\sqrt{n}}\sum_{t=p+k+1}^{n}\eta_{t}\eta_{t-k} + \dfrac{1}{\sqrt{n}}\sum_{t=p+k+1}^{n}A_{1nt}+\dfrac{1}{\sqrt{n}}\sum_{t=p+k+1}^{n}A_{2nt} + \dfrac{1}{\sqrt{n}}\sum_{t=p+k+1}^{n}A_{3nt},  \\
		\sigma^2_\xi\sqrt{n}\widetilde{\gamma}_{k} &= \dfrac{1}{\sqrt{n}}\sum_{t=p+k+1}^{n}\xi_t\xi_{t-k} + \dfrac{1}{\sqrt{n}}\sum_{t=p+k+1}^{n}B_{1nt}+\dfrac{1}{\sqrt{n}}\sum_{t=p+k+1}^{n}B_{2nt} + \dfrac{1}{\sqrt{n}}\sum_{t=p+k+1}^{n}B_{3nt},
		\end{align*}
		where $\xi_t=|\eta_t|-\tau_2$, 
		\[A_{1nt}=\left(\widehat{\eta}_{t}-\eta_{t}\right)\eta_{t-k}, A_{2nt}=\eta_{t}\left(\widehat{\eta}_{t-k}-\eta_{t-k}\right), A_{3nt}=\left(\widehat{\eta}_{t}-\eta_{t}\right)(\widehat{\eta}_{t-k}-\eta_{t-k});\]
		\[B_{1nt}=\left(|\widehat{\eta}_{t}|-|\eta_t|\right)\xi_{t-k}, B_{2nt}=\xi_{t}(|\widehat{\eta}_{t-k}|-|\eta_{t-k}|), B_{3nt}=\left(|\widehat{\eta}_{t}|-|\eta_{t}|\right)(|\widehat{\eta}_{t-k}|-|\eta_{t-k}|).\]
		
		Note that $\partial\eta_t(\bm\theta)/\partial\bm\theta=(-\bm Y_{t-1}^\prime/(\bm\beta^\prime \bm X_{t-1}),-(y_t-\bm \alpha^\prime \bm Y_{t-1})\bm X_{t-1}^\prime/(\bm\beta^\prime \bm X_{t-1})^2)^{\prime}$ and $\sqrt{n}(\widehat{\bm\theta}_{n} - \bm\theta_{0})=O_p(1)$. Then by Taylor's expansion and the ergodic theorem, we have 
		\begin{align}\label{A1n}
		\dfrac{1}{\sqrt{n}}\sum_{t=p+k+1}^{n}A_{1nt}=&\dfrac{1}{n}\sum_{t=p+k+1}^{n}\eta_{t-k}\dfrac{\partial\eta_t(\bm \theta_0)}{\partial \bm \theta^\prime}\sqrt{n}(\widehat{\bm \theta}_n-\bm \theta_0)+o_p(1) \nonumber\\
		=& -\dfrac{1}{n}\sum_{t=p+k+1}^{n}\dfrac{\eta_{t-k}\bm Y_{t-1}^\prime}{\bm \beta_{0}^\prime \bm X_{t-1}}\sqrt{n}(\widehat{\bm\alpha}_n-\bm\alpha_{0}) \nonumber\\
		&-\dfrac{1}{n}\sum_{t=p+k+1}^{n}\dfrac{\eta_{t-k}\eta_t\bm X_{t-1}^\prime}{\bm \beta_{0}^\prime \bm X_{t-1}}\sqrt{n}(\widehat{\bm\beta}_n-\bm\beta_{0})+o_p(1) \nonumber\\
		=&\bm U_{\rho k}\sqrt{n}(\widehat{\bm \theta}_n-\bm \theta_0)+o_p(1).
		\end{align}
		Similarly, it can be shown that
		\begin{align}\label{A2nA3n}
		\dfrac{1}{\sqrt{n}}\sum_{t=p+k+1}^{n}A_{2nt}&=\dfrac{1}{n}\sum_{t=p+k+1}^{n}\eta_{t} \sqrt{n}\left[\eta_{t-k}(\widehat{\bm \theta}_{n})-\eta_{t-k}(\bm\theta_0)\right] = o_{p}(1), \; \text{and} \nonumber\\
		\dfrac{1}{\sqrt{n}}\sum_{t=p+k+1}^{n}A_{3nt}&=\dfrac{1}{n} \sum_{t=p+k+1}^{n}(\widehat{\eta}_{t}-\eta_{t})\sqrt{n}\left[\eta_{t-k}(\widehat{\bm \theta}_{n})-\eta_{t-k}(\bm\theta_0)\right] =o_{p}(1). 
		\end{align}
		As a result, we have
		\begin{equation}\label{rhok rep}
		\sqrt{n}\widetilde{\rho}_{k} = \dfrac{1}{\sqrt{n}}\sum_{t=p+k+1}^{n}\eta_{t}\eta_{t-k} + \bm U_{\rho k}\sqrt{n}(\widehat{\bm \theta}_n-\bm \theta_0)+o_p(1).
		\end{equation}
		Similar to the proof of \eqref{A1n}, we can show that 
		\[
		\dfrac{1}{\sqrt{n}}\sum_{t=p+k+1}^{n}B_{1nt}=\bm U_{\gamma k}\sqrt{n}(\widehat{\bm \theta}_n-\bm \theta_{0})+o_p(1).
		\]
		Note that $E(\xi_t)=0$. Similar to the proof of \eqref{A2nA3n}, we have
		$n^{-1/2}\sum_{t=p+k+1}^{n}B_{2nt}=o_p(1)$ and $n^{-1/2}\sum_{t=p+k+1}^{n}B_{3nt}=o_p(1)$. 
		Therefore, it holds that
		\begin{equation}\label{gammak rep}
		\sqrt{n}\widetilde{\gamma}_{k} = \dfrac{1}{\sqrt{n}}\sum_{t=p+k+1}^{n}\dfrac{\xi_{t}\xi_{t-k}}{\sigma^2_\xi} + \dfrac{\bm U_{\gamma k}}{\sigma^2_\xi}\sqrt{n}(\widehat{\bm \theta}_n-\bm \theta_0)+o_p(1).
		\end{equation}
		Combining \eqref{rhok rep} and \eqref{gammak rep}, together with \eqref{QMLE representation} and \eqref{relationship between hat_acf and tilde_acf}, we can obtain that
		\begin{equation}\label{relationship between hat_acf and true_acf}
		\sqrt{n}(\widehat{\bm \rho}^\prime,\widehat{\bm \gamma}^\prime)^\prime=V\dfrac{1}{\sqrt{n}}\sum_{t=p+k+1}^{n}\bm v_t+o_p(1),
		\end{equation}
		where $\bm v_t=\left(\eta_t\eta_{t-1},\ldots,\eta_{t}\eta_{t-M},\xi_{t}\xi_{t-1}/{\sigma}_\xi^{2},\ldots,\xi_{t}\xi_{t-M}/{\sigma}_\xi^{2},
		\partial \ell_t(\bm \theta_{0})/\partial \bm\theta^\prime\Sigma^{-1}\right)^\prime$ and
		\[
		V=\left(\begin{array}{ccc}
		I_{M} & 0 & U_{\rho} \\
		0 & I_{M} & U_{\gamma}/{\sigma}_\xi^{2}
		\end{array}\right).
		\]
		Then by the martingale central limit theorem and the Cram\'{e}r-Wold device, we have
		\begin{equation*}
		\sqrt{n}(\widehat{\bm\rho}^{\prime},\widehat{\bm\gamma}^{\prime})^{\prime}\rightarrow_{\mathcal{L}} N\left(0, V G V^{\prime}\right),
		\end{equation*}
		where $G=E(\bm v_{t} \bm v_{t}^{\prime})$. The proof of this theorem is accomplished. 
	\end{proof}

	\newpage
	\bibliography{AsymmetricLDAR}
	
	\clearpage
	\begin{figure}
		\centering
		\includegraphics[width=6.4in]{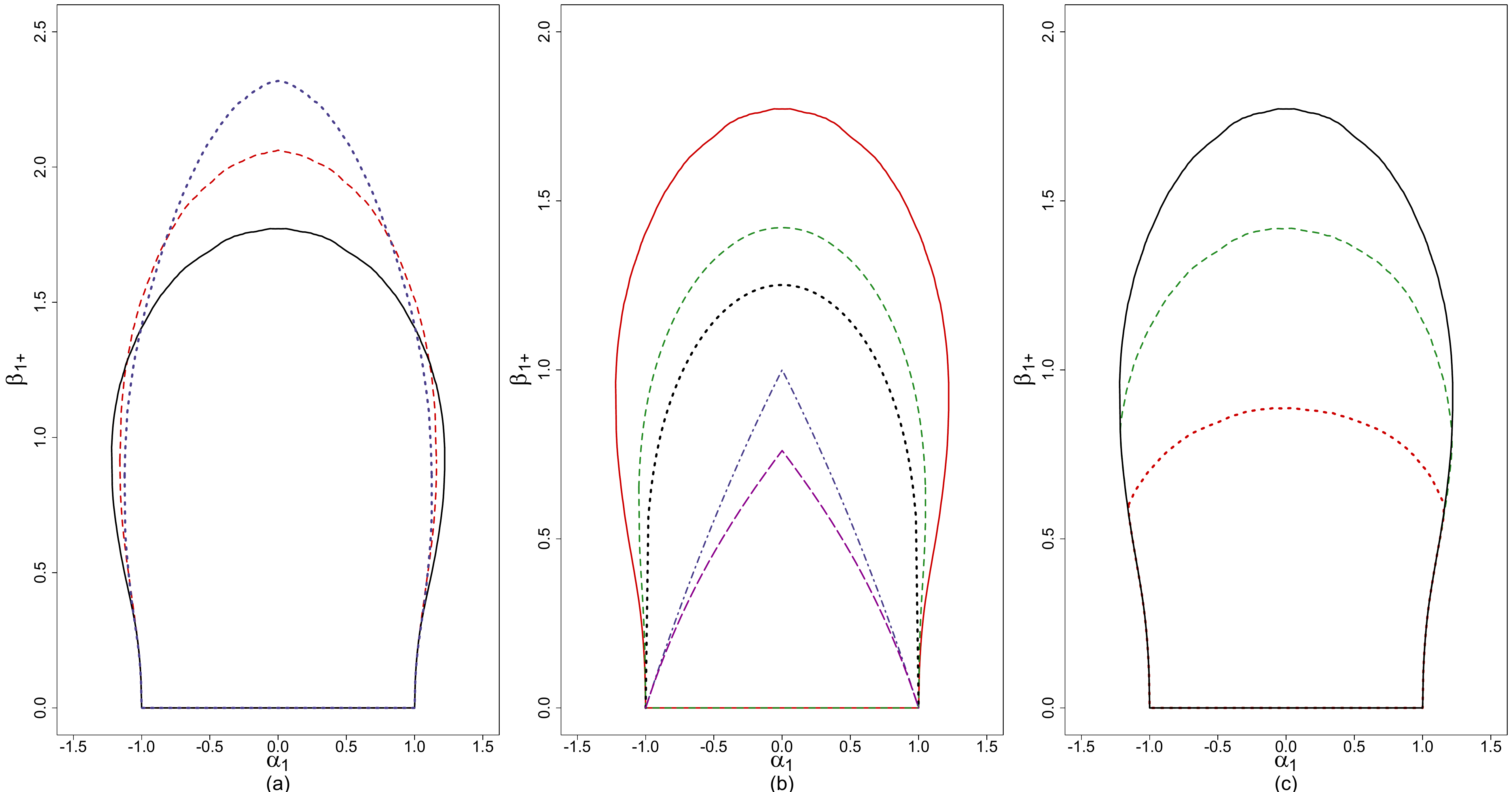}
		\caption{\label{fig1}  Stationarity regions of model \eqref{model} of order one.
			(a): $\beta_{1+}=\beta_{1-}$, $\kappa=0.1$ and $\eta_t$ follows the standard normal $N(0,1)$ (black solid line), Student's $t_5$ (red dashed line) or Laplace (blue dotted line) distribution.
			(b): $\beta_{1+}=\beta_{1-}$, $\eta_{t}\sim N(0,1)$ and $\kappa=0.1$ (red solid line), $0.6$ (green dashed line), $1$ (black dotted line), $2$ (blue dotdash line) or $4$ (purple longdash line).
			(c): $\eta_t\sim N(0,1)$, $\kappa$=0.1, and $\beta_{1+}=d\beta_{1-}$ with $d=0.5$ (red doted line), $d=0.8$ (green dashed line) and $d=1$ (black solid line). }
	\end{figure}

	\begin{table}
		\caption{\label{tableQMLE} Biases ($\times 10$), ESDs ($\times 10$) and ASDs ($\times 10$) of the QMLE $\widehat{\bm \theta}_n$ when the innovations follow the standard normal, standardized Student $t_{5}$ or standardized skewed Student $st_{5,-1.2}$ distribution.}
		\begin{center}
			\begin{tabular}{crrrrrrrrrrrr}
				\hline\hline
				&&\multicolumn{3}{c}{$N(0,1)$}&&\multicolumn{3}{c}{$t_5$}&&\multicolumn{3}{c}{$st_{5,-1.2}$}\\
				\cline{3-5}\cline{7-9}\cline{11-13}
				&$n$&\multicolumn{1}{c}{Bias}&\multicolumn{1}{c}{ESD}&\multicolumn{1}{c}{ASD}&&\multicolumn{1}{c}{Bias}&\multicolumn{1}{c}{ESD}&\multicolumn{1}{c}{ASD}&&\multicolumn{1}{c}{Bias}&\multicolumn{1}{c}{ESD}&\multicolumn{1}{c}{ASD}\\
				\hline
				$\alpha$ &	$500$	&	-0.024	&	0.515	&	0.520	&&	-0.022	&	0.570	&	0.548	&&	-0.024	 &	0.568	&	0.554\\
				&	$1000$	&	 0.022	&	0.371	&	0.369	&&	-0.019	&	0.378	&	0.392	&&	-0.001	 &  0.401	&	0.395\\
				&	$2000$  &	-0.022	&	0.257	&	0.262	&&	-0.004	&	0.274	&	0.278	&&	-0.014   &	0.291	&	0.280\\
				$\omega$ &	$500$	&	 0.040	&	0.271	&	0.264	&&	0.017	&	0.426	&	0.386	&&	 0.009	 &	0.427	&	0.403\\
				&	$1000$	&	 0.016	&	0.193	&	0.187	&&	0.016	&	0.324	&	0.290	&&	 0.008	 &  0.333	&	0.311\\
				&	$2000$  &	 0.010	&	0.129	&	0.132	&& -0.001	&	0.214	&	0.213	&&	-0.003   &	0.231	&	0.223\\
				$\beta^+$  &	$500$	&	-0.088	&	0.640	&	0.587	&& -0.013   &	1.049	&	0.957	&&  -0.039   &	1.191	&	1.031\\
				&	$1000$	&	-0.051	&	0.418	&	0.416	&& -0.006   &	0.819	&	0.724	&&  -0.022	 &	0.860	&	0.793\\
				&	$2000$  &	-0.020	&	0.295	&	0.295	&&  0.001	&	0.565	&	0.534	&&  -0.008	 &	0.619	&	0.570\\
				$\beta_-$  &	$500$	&	-0.147	&	0.736	&	0.699	&& -0.109   &	1.262	&	1.145	&&  -0.156	 &	1.274	&	1.182\\
				&	$1000$	&	-0.045	&	0.509	&	0.498	&& -0.059	&	0.896	&	0.860	&&  -0.075	 &	0.962	&	0.914\\
				&	$2000$  &	-0.037	&	0.352	&	0.352	&& -0.037	&	0.687	&	0.636	&&  -0.049   &  0.725	&	0.660\\
				\hline
			\end{tabular}
		\end{center}
	\end{table}

	\begin{table}
		\caption{\label{tableBIC} Percentages of underfitted, correctly selected and overfitted models by BIC$_1$ and BIC$_2$ when the innovations follow the standard normal, standardized Student $t_{5}$ or standardized skewed Student $st_{5,-1.2}$ distribution.}
		\begin{center}
			\begin{tabular}{lrrrrrrrrrrrr}
				\hline\hline
				&&\multicolumn{3}{c}{$N(0,1)$}&&\multicolumn{3}{c}{$t_5$}&&\multicolumn{3}{c}{$st_{5,-1.2}$}\\
				\cline{3-5}\cline{7-9}\cline{11-13}	
				& $n$	&	Under &	Exact	& Over && Under &	Exact	& Over &&	Under &	Exact	& Over\\\hline
				BIC$_1$ &  200 &0   &49.9&50.1 &&0.3 &55.2&44.5 &&  0.6 & 53.6&45.8\\
				&  500 &0.1 &94.7&5.2  &&2.1 &91.4&6.5  &&  2.2 & 91.1&6.7\\
				& 1000 &0   &100 &0    &&1.8 &98.1&0.1  &&  2.0 & 98.0&0\\
				BIC$_2$ &  200 &1.1 &89.8&9.1  && 2.7&88.3&9.0  && 3.8 &86.6& 9.6 \\
				&  500 &0.4 &99.5&0.1  && 5.1&94.6&0.3  && 5.4 &93.6& 0.1 \\
				& 1000 &0   &100 &0    && 4.4&95.6&0    && 6.4 &93.6& 0 \\
				\hline
			\end{tabular}	
		\end{center}
	\end{table}
	
	\begin{table} 
		\caption{\label{tablesize} Empirical sizes of three tests $W_n$, $L_n$ and $Q_n$ at $5\%$ significance level, where the innovations follow the standard normal, standardized Student $t_{5}$ or standardized skewed Student $st_{5,-1.2}$ distribution.}	
		\begin{center}
			\begin{tabular}{rrrrrrrrrrrr}
				\hline\hline
				&\multicolumn{3}{c}{$N(0,1)$}&&\multicolumn{3}{c}{$t_5$}&&\multicolumn{3}{c}{$st_{5,-1.2}$}\\
				\cline{2-4}\cline{6-8}\cline{10-12}
				$n$&\multicolumn{1}{c}{$W_n$}&\multicolumn{1}{c}{$L_n$}&\multicolumn{1}{c}{$Q_n$}&&\multicolumn{1}{c}{$W_n$}&\multicolumn{1}{c}{$L_n$}&\multicolumn{1}{c}{$Q_n$}&&\multicolumn{1}{c}{$W_n$}&\multicolumn{1}{c}{$L_n$}&\multicolumn{1}{c}{$Q_n$}\\
				\hline
				$500$   &	0.062	&	0.054	&	0.061	&&  0.065	&	0.038	&	0.063	&&	0.061	 &	0.043	&	0.060\\
				$1000$  &	0.061	&	0.058	&	0.059	&&	0.056	&	0.041	&	0.058	&&	0.060	 &	0.040	&	0.059\\
				$2000$  &	0.048	&	0.047	&	0.047	&&	0.052   &	0.047	&	0.053	&&	0.053    &	0.048	&	0.054\\
				\hline
			\end{tabular}
		\end{center}
	\end{table}
	
	\begin{figure}
		\centering
		\includegraphics[width=6in]{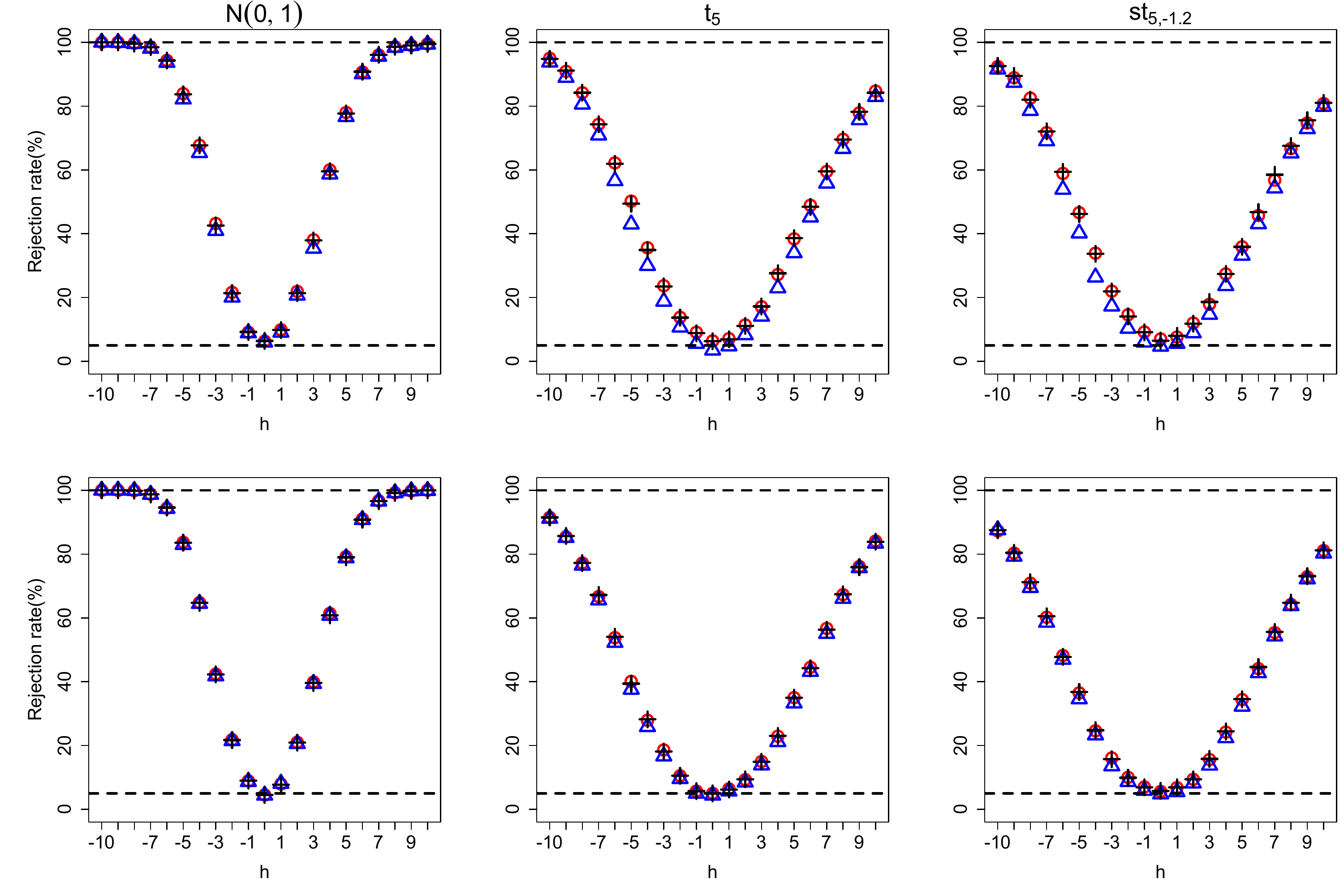}
		\caption{\label{fig2} Local power comparison at the $5\%$ significance level. Upper: $n=500$; bottom: $n=2000$. Circle($\circ$): $W_n$; triangle($\triangle$): $L_n$; cross($+$): $Q_n$.}
	\end{figure}
	
	\begin{table}
		\caption{\label{tablechecking} Rejection rates of the tests $Q(6)$ at the 5\% significance level, where the innovations follow the standard normal, standardized $t_{5}$ or standardized $st_{5,-1.2}$ distribution.}	
		\begin{center}	
			\begin{tabular}{rrccccccccccc}
				\hline\hline
				&&\multicolumn{3}{c}{$N(0,1)$}&&\multicolumn{3}{c}{$t_5$}&&\multicolumn{3}{c}{$st_{5,-1.2}$}\\
				\cline{3-5}\cline{7-9}\cline{11-13}
				$c_1$ & $c_2$ &500&1000&2000&&500&1000&2000&&500&1000&2000\\
				\hline		
				0.0	& 0.0 &  0.059 & 0.057 & 0.050 && 0.067 & 0.062 & 0.057 && 0.061 & 0.058 & 0.057 \\
				0.1 & 0.0 &  0.214 & 0.471 & 0.855 && 0.240 & 0.440 & 0.806 && 0.222 & 0.433 & 0.827 \\
				0.3 & 0.0 &  0.997 & 1.000 & 1.000 && 0.994 & 1.000 & 1.000 && 0.999 & 1.000 & 1.000 \\
				0.0 & 0.1 &  0.085 & 0.201 & 0.455 && 0.072 & 0.120 & 0.276 && 0.076 & 0.125 & 0.264 \\
				0.0 & 0.3 &  0.655 & 0.989 & 1.000 && 0.428 & 0.862 & 0.997 && 0.418 & 0.886 & 0.999 \\
				\hline
			\end{tabular}
		\end{center}
	\end{table}		
	
	\begin{figure}[htp]
		\centering
		\includegraphics[width=6.4in]{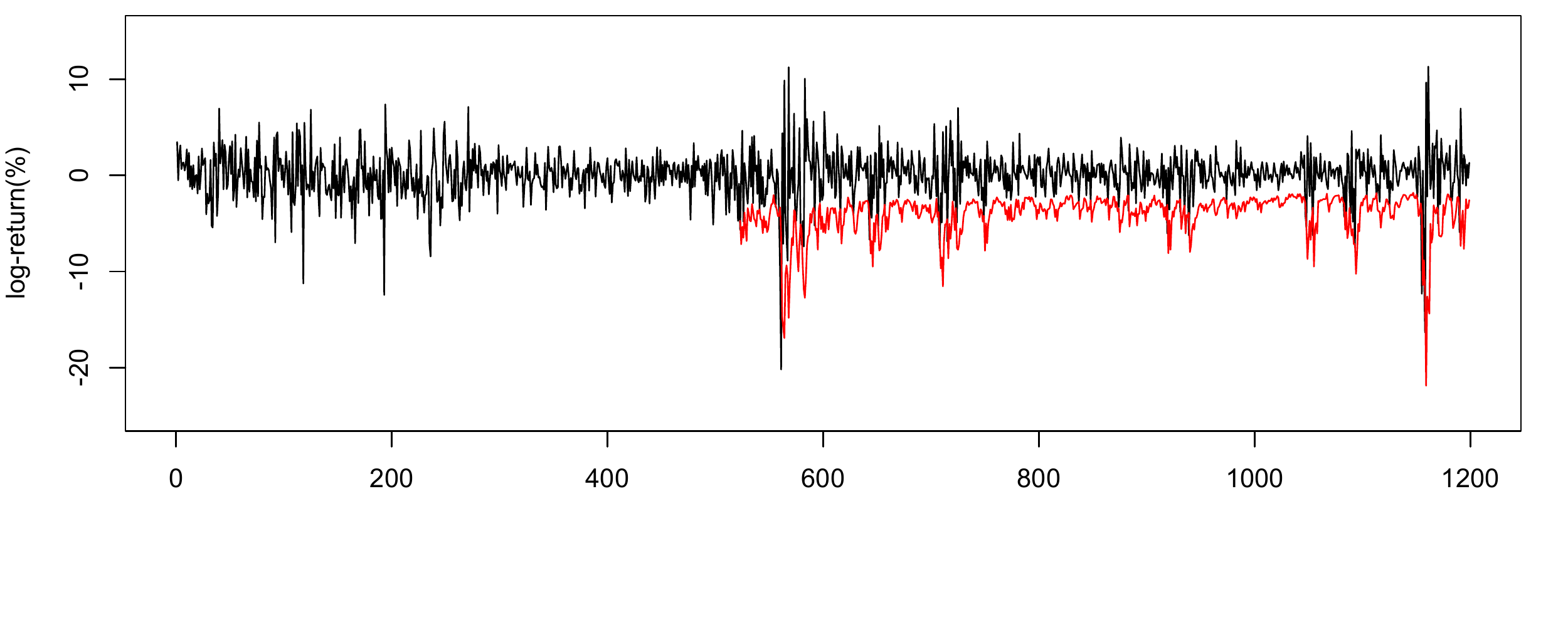}
		\caption{\label{fig_real_data}  Time plot for centered weekly log returns in percentage (black line) of S\&P500 Index from January 1998 to December 2021, with one-week negative VaR forecasts at the level of 5\% (red line) from January 2008 to December 2020.}
	\end{figure}
	
	\begin{table}
		\caption{\label{table_real_data_stat} Summary statistics for S\&P500 returns.}
		\begin{center}
			\begin{tabular}{ccccccc}
				\hline \hline
				Mean& Median & Std.Dev. & Skewness & Kurtosis & Min & Max\\
				\hline
				0.00 & 0.12 & 2.53 & -0.90 & 7.17&-20.20&11.30\\
				\hline
			\end{tabular}
		\end{center}
	\end{table}
	
	\begin{figure}[htp]
		\centering
		\includegraphics[width=6.4in]{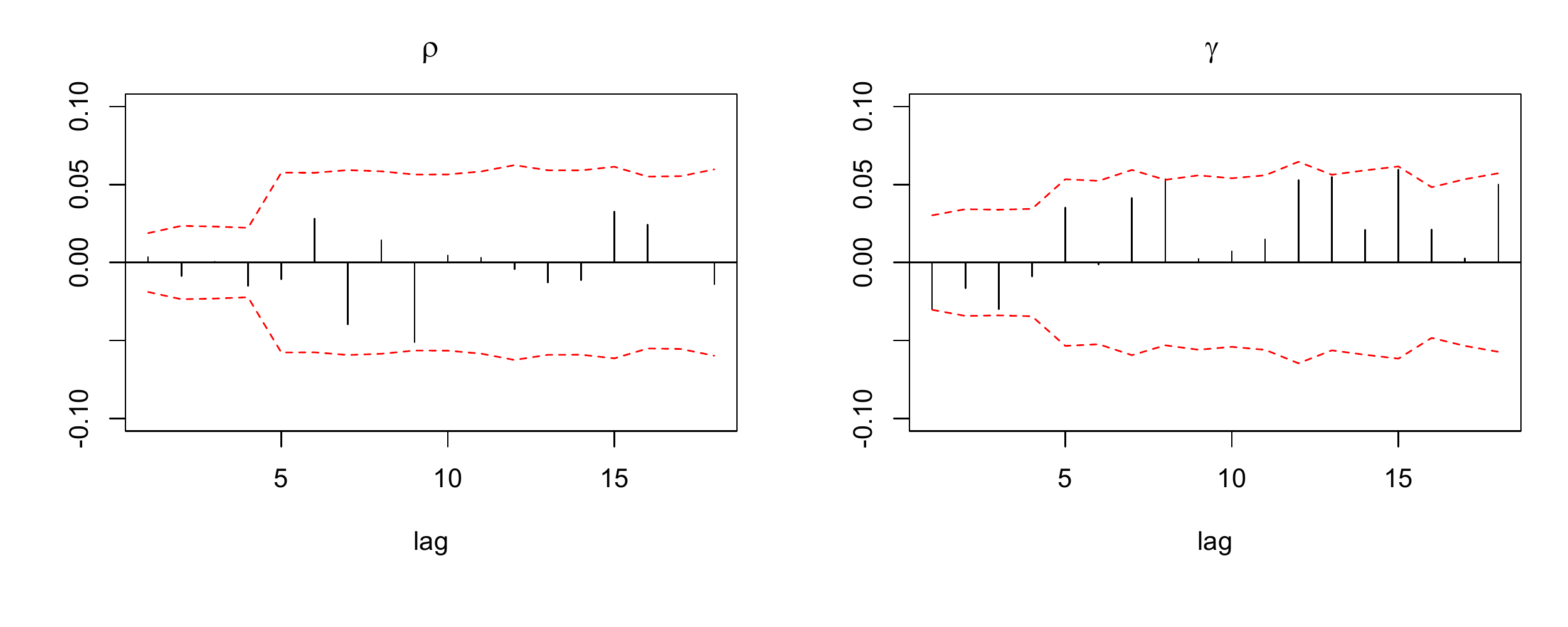}
		\caption{\label{fig_real_data_portmanteau_test}  Residual ACF plots for $\widehat{\rho}_l$ (left panel) and $\widehat{\gamma}_l$ (right panel), where the red dashed lines are the corresponding 95\% confidence bounds. }
	\end{figure}
	
	\begin{table}
		\caption{\label{table_real_data_forecast} Empirical coverage rate (\%) and $p$-values of two VaR backtests of three models at the 1\%, 5\%, 95\% and 99\% conditional quantiles. M1, M2 and M3 represent the ALDAR, LDAR and AR-TGARCH models, respectively. }
		\begin{center}
			\begin{tabular}{crrrrrrrrrrrrrrr}
				\hline\hline 
				& \multicolumn{3}{c}{$\tau=1\%$} && \multicolumn{3}{c}{$\tau=5\%$} && \multicolumn{3}{c}{$\tau=95\%$} && \multicolumn{3}{c}{$\tau=99\%$}\\
				\cline{2-4}\cline{6-8}\cline{10-12}\cline{14-16}
				&\multicolumn{1}{c}{ECR}&\multicolumn{1}{c}{CC}&\multicolumn{1}{c}{DQ}&&\multicolumn{1}{c}{ECR}&\multicolumn{1}{c}{CC}&\multicolumn{1}{c}{DQ}&&\multicolumn{1}{c}{ECR}&\multicolumn{1}{c}{CC}&\multicolumn{1}{c}{DQ}&&\multicolumn{1}{c}{ECR}&\multicolumn{1}{c}{CC}&\multicolumn{1}{c}{DQ}\\
				\hline
				M1	   &	1.48	&	0.16	&	0.10	&&	5.76	&	0.58	&	0.89	&&	94.68	 &	0.93	&	0.98	&&	98.82	 &	0.82	&	0.96\\
				M2	   &	3.54	&	0.00	&	0.00	&&	10.04	&	0.00	&	0.00	&&	88.18	 &  0.00	&	0.00	&&	93.80	 &	0.00	&	0.00\\
				M3 	   &	1.33	&	0.19	&	0.06	&&	5.02	&	0.26	&	0.33	&&	94.68    &	0.35	&	0.06	&&	98.22	 &	0.15	&	0.02\\
				\hline
			\end{tabular}
		\end{center}
	\end{table}
	
\end{document}